\documentclass[fleqn,reqno,11pt,a4paper,final]{amsart}

\usepackage[a4paper,left=30mm,right=30mm,top=30mm,bottom=30mm,marginpar=20mm]{geometry}
\usepackage{amsmath}
\usepackage{amssymb}
\usepackage{amsthm}
\usepackage{amscd}
\usepackage[ansinew]{inputenc}
\usepackage[noadjust]{cite}
\usepackage{bbm}
\usepackage{color}
\usepackage[english=american]{csquotes}
\usepackage[final]{graphicx}
\usepackage{hyperref}
\usepackage{calc}
\usepackage{mathptmx}
\usepackage{tikz}
\usepackage{hyperref}

\graphicspath{{../Pictures/}}

\numberwithin{equation}{section}

\newtheoremstyle{thmlemcorr}{10pt}{10pt}{\itshape}{}{\bfseries}{.}{10pt}{{\thmname{#1}\thmnumber{ #2}\thmnote{ (#3)}}}
\newtheoremstyle{thmlemcorr*}{10pt}{10pt}{\itshape}{}{\bfseries}{.}\newline{{\thmname{#1}\thmnumber{ #2}\thmnote{ (#3)}}}
\newtheoremstyle{remexample}{10pt}{10pt}{}{}{\bfseries}{.}{10pt}{{\thmname{#1}\thmnumber{ #2}\thmnote{ (#3)}}}
\newtheoremstyle{ass}{10pt}{10pt}{}{}{\bfseries}{.}{10pt}{{\thmname{#1}\thmnumber{ A#2}\thmnote{ (#3)}}}

\theoremstyle{thmlemcorr}
\newtheorem{theorem}{Theorem}
\numberwithin{theorem}{section}
\newtheorem{lemma}[theorem]{Lemma}

\newtheorem{definition}[theorem]{Definition}

\theoremstyle{thmlemcorr*}
\newtheorem{theorem*}{Theorem}
\newtheorem{lemma*}[theorem]{Lemma}
\newtheorem{corollary*}[theorem]{Corollary}
\newtheorem{proposition*}[theorem]{Proposition}
\newtheorem{problem*}[theorem]{Problem}
\newtheorem{conjecture*}[theorem]{Conjecture}
\newtheorem{definition*}[theorem]{Definition}
\newtheorem{assumption*}[theorem]{Assumption}

\theoremstyle{remexample}

\theoremstyle{ass}
\newtheorem{assumption}{Assumption}

\newcommand{\R}{\mathbb{R}}

\def\XXint#1#2#3{{\setbox0=\hbox{$#1{#2#3}{\int}$}
\vcenter{\hbox{$#2#3$}}\kern-.5\wd0}}

\renewcommand{\epsilon}{\varepsilon}
\renewcommand{\phi}{\varphi}

\begin{document}


\title[Matching Code and Law: Achieving Algorithmic Fairness with Optimal Transport]{Matching Code and Law: Achieving Algorithmic Fairness with Optimal Transport}

\author{Meike Zehlike}
\address{\textit{Meike Zehlike:} Max Planck Institute for Software Systems,
Campus E1~5, 66123 Saarbr\"ucken, Germany}
\email{meikezehlike@mpi-sws.de}

\author{Philipp Hacker}
\address{\textit{Philipp Hacker:} Law Department, Humboldt University of Berlin, Unter den Linden~6, 10099 Berlin, Germany; A.SK Fellow, WZB Social Sciences Centre; Research Fellow, Centre for Law, Economics and Society/Centre for Blockchain Technologies, UCL}
\email{philipp.hacker@rewi.hu-berlin.de}

\author{Emil Wiedemann}
\address{\textit{Emil Wiedemann:} Institute of Applied Analysis, Ulm University, Helmholtzstr.~18, 89081 Ulm, Germany}
\email{emil.wiedemann@uni-ulm.de}

\begin{abstract}

Increasingly, discrimination by algorithms is perceived as a societal and legal problem. As a response, a number of criteria for implementing algorithmic fairness in machine learning have been developed in the literature. This paper proposes the Continuous Fairness Algorithm (CFA$\theta$) which enables a continuous interpolation between different fairness definitions. More specifically, we make three main contributions to the existing literature. First, our approach allows the decision maker to continuously vary between specific concepts of individual and group fairness. As a consequence, the algorithm enables the decision maker to adopt intermediate ``worldviews'' on the degree of discrimination encoded in algorithmic processes, adding nuance to the extreme cases of ``we're all equal'' (WAE) and ``what you see is what you get'' (WYSIWYG) proposed so far in the literature. Second, we use optimal transport theory, and specifically the concept of the barycenter, to maximize decision maker utility under the chosen fairness constraints. Third, the algorithm is able to handle cases of intersectionality, i.e., of multi-dimensional discrimination of certain groups on grounds of several criteria. We discuss three main examples (credit applications; college admissions; insurance contracts) and map out the legal and policy implications of our approach. The explicit formalization of the trade-off between individual and group fairness allows this post-processing approach to be tailored to different situational contexts in which one or the other fairness criterion may take precedence. Finally, we evaluate our model experimentally.

\end{abstract}







\maketitle




\section{Introduction}
Suppose a decision needs to be made for a large number of individuals that fall into different groups. We call the agent or institution taking this decision the \emph{decision maker}. As a motivating example, imagine the decision maker to be a provider of credit scores which are subsequently used by financial institutions to decide on credit applications. The score indicates whether an applicant should be made a credit offer at all, and if so under what conditions. There is ample evidence, however, which suggests that credit scores are ridden with bias, unfairly discriminating against minorities (Pasquale 2015; Rothmann et al. 2014; Hurley and Adebayo 2016). Simultaneously, providers of credit scores growingly harness machine learning technology to predict credit risk for individuals (cf. Fuster et al. 2018; German Federal Ministry of Justice 2018; Hurley and Adebayo 2016). In fact, in many other scenarios involving a large number of subjects and data, decision makers increasingly employ algorithmic methods to reach their decisions. Other high-stakes settings in which algorithmic decision making is increasingly prevalent include hiring or insurance decisions, as we discuss in more detail below, or criminal justice decisions (see Berk et al., 2018). We therefore seek to establish a novel framework that subjects the decision maker to fairness constraints which can be operationalized in a setting of algorithmic decision making;\footnote{ We would like to point out that our framework also covers cases of non-algorithmic decision making and thus applies widely to decisions implicating fairness.} these constraints facilitate a trade-off between specific concepts  of individual and group fairness (thereby minimizing bias in algorithmic decision making) while preserving as much utility for the decision maker as possible.

For each individual in our credit scoring example, certain \emph{observable data} is available, like credit history, degree certificates, scores quantitatively reflecting an interviewer's judgment, etc. Part of this data might have been collected or processed through procedures biased against certain groups of individuals. Indeed, evidence is mounting that bias generally haunts algorithmic decision procedures, particularly those based on machine learning methods (cf. Barocas and Selbst 2016; Reuters 2018). Moreover, even if the data was handled correctly, it might allow for inferences about an individual's ethnicity, gender, or other criteria about membership in legally protected groups. In other words, in a number of settings, full disclosure of the observable data to the decision maker is likely to result in \emph{unfair} decisions.

Instead, in our framework an \emph{impartial agent}, like a regulatory body, will transform the observable ``raw data" into a \emph{fair representation} of the data (cf. Zemel et al. 2013 for the terminology), which can then be used by the decision maker to take its decision. For instance, if the observable data is biased against one group of individuals, the data would be ``corrected" by being transported to a fair representation that mitigates, or even fully removes, this bias.

The map from an individual's observable data to the fair representation is the main subject of this work. It should satisfy, as closely as possible, a number of requirements reflecting the ``fairness" of the process while maintaining decision maker utility to a maximum:
\begin{enumerate}
\item\label{individual} Individual fairness: Similar observable data should be mapped to similar fair data.
\item\label{group} Group fairness: The fair data of an individual should not allow for any inference on the individual's group membership (statistical parity).
\item\label{monotone} Monotonicity: If one individual scores better in the raw data, then they should also score better, or equally well, in the fair data. 
\item\label{utility} Decision maker utility: The fair data should conceal as little information as possible to the decision maker.
\end{enumerate}         
These requirements will be refined, and defined more rigorously, below (see Section~\ref{model}). They have been well-studied except for monotonicity, and it is known that they cannot all be perfectly satisfied at the same time (Friedler et al. 2016); this is a general problem in fairness constraints in algorithmic decision making  (Chouldechova 2017; Kleinberg et al. 2016). For instance, if groups have different raw score distributions, it will be necessary to attribute different fair scores to two individuals with similar raw scores if they belong to different groups  in order to achieve group fairness; however, this will violate monotonicity and individual fairness. If, as will typically be the case, different groups exhibit different statistics in the raw data, then indeed the only way to satisfy requirements (\ref{individual})--(\ref{monotone}) simultaneously is to assign the \emph{same} fair score to \emph{every} individual (cf. Kleinberg et al. 2016). This, of course, would neither be ``fair" in any meaningful sense, nor would it be of any use to the decision maker, as they would be completely ``blindfolded" in making their decision. Correspondingly, requirement (\ref{utility}) would be violated to an extreme extent in this case. 

Our general discussion is close to the works of Dwork et al. (2012) and of Feldman et al. (2015). In contrast to these and other fairness approaches, however, our model contains a stringent mathematical proof of optimality. As the model will show in greater detail, we introduce a novel framework which greatly facilitates the mathematical analysis of the ``raw-to-fair" map. One of the main novelties of this paper is thus to work in a \emph{continuous} rather than a discrete setting, in the Euclidean space $\R^n$. We would like to stress at the outset that our framework is continuous in a dual sense. First, our probability measures, i.e., the distributions of (raw and fair) scores across groups, are continuous (absolutely continuous with respect to Lebesgue measure, to be precise). This not only allows for highly fine-grained scoring, but it is also a mathematical prerequisite for the existence of a unique \emph{optimal transport map} that we will use in our model to maximize decision maker utility. The choice of a continuous probability measure is justified as long as the number of individuals in question is sufficiently large, and as long as it is unlikely that many of the individuals have \emph{exactly} the same raw score.\footnote{ In particular this requires that the raw data contains sufficiently fine degrees of evaluation. If this is granted, then adding some stochastic noise to the evaluation of the raw data will remove undesired ``concentration" effects, if necessary.} Continuous distributions are frequently used in economics (Romei and Ruggieri 2014, p. 589) and admit a richer mathematical theory. In this paper, specifically, we will exploit some elements of the mathematical theory of \emph{optimal transport}. Working with continuous distributions is generally not a significant restriction, as every continuous distribution can be arbitrarily approximated by discrete ones in Wasserstein distance (i.e.\ in the weak topology), and vice versa. 

Second, another type of continuity resides in the possibility to variably choose a parameter, $\theta$, that allows us to skew the model more towards individual or toward group fairness, as conditions warrant. Friedler et al. (2016) consider not only the observable data and the decision finally made by the decision maker, but also a so-called \emph{construct space}, which can be thought to contain the ``actual'' or ``true'' properties of every individual. The (possibly inaccurate or biased) evaluation of the ``true'' data is then modelled by a map from construct space to observable space. Construct space and the construct-to-observable map are, by definition, not measurable or observable in any way. Rather, the relation between construct and observable space must be \emph{postulated} in an axiomatic fashion. To this end, the authors of Friedler et al. (2016) propose two extreme ``worldviews'': WAE (``we're all equal'') and WYSIWYG (``what you see is what you get''). In the first case, it is assumed that any differences between groups in the distribution of the raw data  in the observed space are due to discrimination, incorrect data handling, or other exogenous factors, and that these differences should therefore not be visible in construct space; insofar as it is relevant for the decision, different groups are assumed to have the same "true" distributions of scores (which differs from the observable distributions). In the second case, the assumption is that the observable data truly reflects the properties of the individuals and can thus be immediately used for decision making. In this worldview, effects of bias or inaccurate collection of data are either flatly denied or tolerated. This assumption therefore collapses the distinction between construct and observed space. Both assumptions, however, are extreme cases: either groups are postulated to be perfectly equal in construct space, or data in the observed space is assumed to be perfectly correct.     

A main contribution of this paper is to add nuance to these distinct worldviews. We suggest a framework that allows to continuously interpolate between WAE and WYSIWYG. Indeed, in the mathematical theory of optimal transport, the technique of \emph{displacement interpolation} is used to continuously move from one probability distribution to another in a particularly natural way, and we apply this tool in the context of fair representations. To this end, we introduce a parameter, $\theta$, that allows us to continuously move from WAE to WYSIWYG. As we will see, this implies that we can equally move from a maximal fulfillment of individual (WYSIWYG: $\theta$ = 0) to a maximal fulfillment of group fairness (WAE: $\theta$ = 1). The resulting Continuous Fairness Algorithm (CFA$\theta$) allows us to formalize the trade-off between these fairness concepts, and therefore to adapt the framework on a case-by-case basis to different decision making contexts in which different fairness constraints may be normatively desirable (on this desideratum, see Binns 2018, p. 6-7). We would like to stress that our model is able to guarantee monotonicity within groups, but not for members belonging to different groups; this is also the source of the violation of individual fairness in our transportation exercise. Furthermore, since our model operates with a continuous distribution, the need for arbitrary randomization in order to achieve a fair distribution is minimized.\footnote{ This is due to the fact that a necessary precondition for the need to randomize is that different individuals have exactly the same fair score; this is excluded in a continuous setting, and can be neglected in the discrete setting if the evaluation procedure is sufficiently fine-grained, see the discussion above.}

Quite obviously, one key question for the construction of a fair representation is the choice of the target representation for $\theta$ = 1: among all possible target measures, one ought to find the one that optimizes individual fairness and decision maker utility. Indeed, the requirement, stemming from the desire to achieve statistical parity, that all the raw score distributions of the different groups be mapped onto one single representation (in this extreme case), does not say anything about what this distribution should look like. Dwork et al. (2012, p. 221) choose the distribution of the privileged group as the target distribution. By contrast, we show that there is a potentially more convincing target distribution which occupies a ``middle ground" between the distributions of all the different groups (the so-called \emph{barycenter} with respect to Wasserstein-2 distance\footnote{ In the one-dimensional case, a kind of barycenter has been used by Feldman et al. (2015), but with respect to the Wasserstein-1 distance.}). Choosing the barycenter as the target distribution has two important advantages: first, it does not impose the distribution of one ``privileged" or majority group onto the other groups. Second, it is the distribution that is closest to all the raw distributions in a least square sense; therefore, it preserves decision maker utility to a maximum if we use optimal transport theory to map all raw distributions onto the barycenter. The \emph{existence} of such an intermediate distribution, the barycenter, is mathematically highly nontrivial and was only proved in 2011 by Agueh and Carlier (2011) under certain assumptions that, fortunately, are plausible for many scenarios of algorithmic decision making.\footnote{ These conditions include absolute continuity of the raw distribution, which can be arbitrarily approximated by discrete distributions in Wasserstein space, as noted above; and a quadratic cost (or utility) function, a condition that can be fulfilled by initially transforming the utility function of the decision maker appropriately.} 

The approach that is probably closest to ours is Feldman et al. (2015, p. 264) (see also the discussion in Section~4). In contrast to their paper, which uses a "median" distribution (i.e.\ the barycenter in the Wasserstein-1 metric), our approach introduces the barycenter in Wasserstein-2 (which operates with the notion of least squares in Wasserstein-2 distance), and uses the unique optimal transport map from raw to fair scores. Quadratic cost and distance functions seem plausible if one assumes that the utility of output data declines ever more sharply (and not only linearly) the more it is removed from a sufficiently precise approximation of the truth; in other words, if the distance between an output and ground truth becomes too large, the data is almost of no use to the decision maker any more. This seems realistic particularly when important differential consequences are attached to different outputs (as in our examples); arguably, it is only in such cases of important distinctions that a decision maker will resort to an algorithmic model in the first place. The barycenter allows us not only to maximize decision maker utility under fairness constraints, but also to vary between different fairness measures (individual vs. group/WYSIWYG vs.\ WAE), depending on the concrete decision-making framework our algorithm is applied to. Further important differences between our approach and the one adopted in Feldman et al. (2015) are that we are able to handle high-dimensional raw scores; and cases of intersectionality (see next paragraph).

To summarize, our CFA$\theta$ has a number of advantageous features, but also some limitations. Concerning the advantages, first, in the WAE case ($\theta$ = 1), group fairness is fulfilled as well as within-group monotonicity. Second, individual fairness and decision maker utility are optimized by the choice of the barycenter as the target distribution, and the optimal transport toward it. Third, the possibility of choosing $\theta$yields substantial flexibility for the decision maker. Hence, our model allows to implement any intermediate worldview the decision maker may have. It does not impose a \emph{specific} fairness constraint. The decision maker may choose $\theta$ such, for example, that the resulting distribution does not correspond to full group fairness, but still fulfills the so-called 80 percent rule, an important threshold for US disparate impact doctrine (EEOC 2015, Section 4 D.; Barocas and Selbst 2016, p. 701 et seq.) that is increasingly gaining traction in EU anti-discrimination law, too. It requires that the probability of a member of a disprivileged group being positively labeled is at least 80 percent of the respective probability of a member of the privileged group (see Feldman et al. 2015; Zafar et al. 2017). Hence, $\theta$ can be consciously chosen to force compliance with existing anti-discrimination legislation. As the discussion of the legal and policy implications of our model shows, the choice of $\theta$ is a deeply normative one and can be adapted to different situations in which individual or rather group fairness should be the primary goal. Fourth, being a post-processing approach, our model can be applied to any machine learning model and is not constrained to specific types (such as linear models). Even more importantly, this implies that the most efficient predictor can be chosen, only to be repaired if issues of discrimination arise, an issue stressed by Kleinberg et al. (2018). This is particularly relevant when, as in the credit scoring setting, an institution (the bank) receives scores from a third party (the scorer, e.g., FICO) without having access to the training data and the model. The institution can apply our algorithm to the outcome distribution regardless of how it was learned. Fifth, from an economic viewpoint, quadratic cost functions are the best choice, particularly in credit scoring settings (see Fuster et al. 2018, p. 7). Sixth, we are able to handle the problem of \emph{intersectionality}, i.e.\ the phenomenon that an individual may belong to several (protected) groups at the same time. As a recent judgment by the Court of Justice of the European Union has shown,\footnote{ CJEU case C-443/15 \emph{Parris} ECLI:EU:C:2016:897.} intersectionality presents a pressing problem in real-life decisions. As mentioned, in our approach to it, in contrast to Dwork et al. (2012), we do not move the data distribution of a discriminated group to that of the privileged one, but rather we transport the distributions of all groups to respective intermediate representations that are chosen by displacement interpolation (the $\theta$ score). Importantly, we may choose $\theta$ differently for different groups; this allows us to treat particularly disadvantaged groups (for example those that fulfill several protected criteria at the same time) "fairer" (with respect to group fairness) than other ones, if so desired. 

This leads us to the first of two important limitations that would like to stress. First, if a generally disadvantaged group performs exceptionally well (better than the barycenter) in one setting, attaching a high $\theta$ value to that group removes part of this unexpected advantage (by moving the group closer to the -- unexpectedly worse -- barycenter). Therefore, the setting of the values should always be combined with a (at least superficial) performance evaluation of the respective groups in the raw scores. Second, as mentioned before, our model does not work well with groups that are small or have little variance because they depart too far from the continuity assumption. As a consequence, while in mathematical theory, it may guarantee, for a certain $\theta$ value, a certain percentage of members of the specific protected group in the top k positions (like Zehlike et al. 2017), this may not be the case in real-world implementations if group sizes are too small.  This will be explored further in the data-driven experiments.

Importantly, to reiterate, in all of these data corrections by the CFA$\theta$, we take the raw score as a given output of a prior learning task; therefore, the full force of machine learning can be unleashed to calculate the raw score which is only transformed into a fair score after its elaboration. While being a stand-alone procedure, our approach may nevertheless be fruitfully complemented by fairness or data collection/quality constraints applied to the calculation of the raw score itself (see below, Section~\ref{relation}).

The remainder of the paper is organized as follows: Section 2 introduces the mathematical model and establishes its fundamental optimality properties. Readers unfamiliar with mathematical notation may jump right to Section 3 which discusses a number of examples. Section 4 places the model within the broader framework of related work. Section 5 offers detailed legal and policy implications. Section 6 contains a data-driven evaluation. Section 7 concludes.

\section{The Model}\label{model}    
Let $X$ be a set of individuals that may have certain traits indexed by $1,\ldots,N$. We are given a \emph{group membership map} $g:X\to\{0;1\}^N$ whose $i$-th component indicates whether or not an individual carries trait number $i$. This induces a partition of $X$ via
\begin{equation*}
X=\bigcup_{k\in\{0;1\}^N}g^{-1}(k)
\end{equation*} 
into at most $2^N$ groups $X_k:=g^{-1}(k)$ (note that some of these could be empty). Let us call the number of (non-empty) groups $G$.

Various data may be collected from individuals in the process of decision-making. This includes qualitative data such as personal interviews, expert opinions, letters of motivation, etc. Obviously, it should not be the objective of an abstract theory of fairness to design a map from this data to a quantitative ranking that reflects the decision maker's preferences. Rather, we assume such a map as given, and therefore our starting point is a \emph{score function} $S:X\to\R^n$. The score, which may be composed of $n$ partial scores for different categories, is assumed to express the decision maker's utility function in the sense that \emph{if the decision maker had the full observable data at their disposal, then they would always prefer an individual with a higher score in each category to one with a lower one}.\footnote{ In many cases, $n=1$ may be sufficient. Indeed, in most real-life applications, at some point all the information about the individuals must be brought into a linear ordering.}  

The restrictions of $S$ to $X_k$ are denoted $S_k$. On $\R^n$, we introduce probability measures\footnote{ We use the term ``probability measure" in the sense of ``normed measure''; no randomness is insinuated by this terminology.} $\mu$ and $\{\mu_k\}_{k=1,\ldots,G}$ that encode the score distribution within the entire set of individuals, and in the $k$-th group, respectively. This means that $\mu(B)$ is the proportion of individuals with a score contained in a subset $B\subset\R^n$, and likewise $\mu_k(B)$ denotes the proportion of individuals with score in $B$ in the $k$-th group. If we index the individuals themselves by a vector in $\R^n$, i.e.\ $X=\R^n$, then we can write 
\begin{equation*}
w_k:=|X_k|=|g^{-1}(k)|
\end{equation*} 
for the proportion of members of group $k$ in the total population $X$, and 
\begin{equation*}
\mu_k=dx^n\circ S_k^{-1}.
\end{equation*}
where $dx^n$ is the Lebesgue measure (i.e., the Euclidean volume, or, in case $n=1$, the length)\footnote{The pullback measure $dx^n\circ S_k^{-1}$ is defined by $dx^n\circ S_k^{-1}(B):=dx^n(\{x\in\R^n: S_k(x)\in B\})$, which gives the proportion of individuals whose score lies in the set $B$.}
Accordingly, we have 
\begin{equation*}
\mu=dx^n\circ S^{-1}=\sum_{k=1}^Gw_k\mu_k.
\end{equation*}
Our decisive continuity assumption can be stated as follows:
\begin{assumption}\label{contassumption}
The measures $\mu$ and $\mu_k$, $k=1,\ldots,G$, are absolutely continuous with respect to Lebesgue measure and have finite variance, with densities $f\in L^1(\R^n)$ and $f_k\in L^1(\R^n)$.
\end{assumption}

The \emph{fair representations} for each group will be maps $T_k$ from $\R^n$ to $\R^n$. They transform a ``raw score'' into a ``fair score" for group $k$. The resulting fair score for an individual $x\in X_k$ is then given by $T_k(S(x))$. The raw-to-fair map $T_k$ will transport the measure $\mu_k$ to the pushforward measure 
\begin{equation}\label{pushforward}
\nu_k=\mu_k\circ T_k^{-1}
\end{equation}
representing the ``fair score distribution'' for group $k$.

We will ask the converse question: Given a target distribution $\nu_k$, how can we choose $T_k$ so that $\nu_k=\mu_k\circ T_k^{-1}$? Any such map $T_k$ is called a \emph{transport map} from $\mu_k$ to $\nu_k$. In general, there are many such transport maps. In the sequel we will describe how to choose $\nu_k$ and corresponding transport maps $T_k$ in order to guarantee a maximal amount of fairness in the sense of requirements (\ref{individual})--(\ref{utility}), using an algorithm which we call (CFA$\theta$). Once these data have been constructed, the decision maker will be presented the distribution
\begin{equation*}
\bar{\nu}:=\sum_{k=1}^Gw_k\nu_k
\end{equation*} 
and will make a decision based on this distribution; an individual $x\in X_k$ will thus be classified through their fair score $T_k(S(x))$. 

We first need to impose target distributions $\nu_k$. To this end, let us discuss the extreme worldviews WYSIWYG and WAE, which will form the endpoints of our interpolation, within our model:

\subsection*{WYSIWYG}
When the ``raw'' score is deemed a true and fair representation of reality, nothing needs to be done, hence we set $\nu_k=\mu_k$ for all $k=1,\ldots,G$ (thereby maximising decision maker utility), and accordingly $\nu=\mu$. Monotonicity then forces $T_k=id$ for all $k$, thus optimising individual fairness at the same time. Of course, if $\mu_k$ are different, then group fairness is violated.  

\subsection*{WAE}  
Under the hypothesis that any differences between the $\mu_k$ emerge solely from undesired exogenous factors, such as bias, and should be removed, there should exist a single target distribution $\nu$ independent of $k$ so that $T_k$ transports $\mu_k$ to $\nu$. This produces statistical parity, hence group fairness is optimised, whereas the other three requirements will typically be violated. The problem is then to find the common fair distribution $\nu$ and the transport maps $T_k$ minimising these violations.

\begin{definition}[Optimal transport map]
Let $\mu, \nu$ be two probability measures on $\R^n$ with finite variance. An \emph{optimal transport map} is a transport map between $\mu$ and $\nu$ that minimises the cost functional
\begin{equation}\label{cost}
C(\mu,\nu,T):=\int_{\R^n}|x-T(x)|^2d\mu(x)
\end{equation}
among all transport maps from $\mu$ to $\nu$.
\end{definition}
In general, there need not exist any transport map between two probability measures.\footnote{ This is typically the case for discrete measures, so that the notion of transport map needs to be relaxed to transport \emph{plan}. In the context of algorithmic fairness, this leads to the necessity of randomisation, as in Dwork et al. (2012).} If there exists a unique optimal transport map $T$ between $\mu$ and $\nu$ however, then the so-called \emph{Wasserstein distance} between $\mu$ and $\nu$ is given by
\begin{equation}\label{wasser}
W_2(\mu,\nu):= C(\mu,\nu,T)^{1/2}.
\end{equation}
The Wasserstein distance forms a metric on the space of all probability measures on $\R^n$.

Other cost functions than $|\cdot|^2$ can be used, but the quadratic cost function has particularly nice properties and appears most appropriate in view of the applications we are interested in. The following theorem is a seminal result of Brenier (1987, 1991) (see also Theorem 1.26 in Ambrosio and Gigli 2013):
\begin{theorem}\label{brenier}
Let $\mu, \nu$ be probability measures on $\R^n$ with finite variance. If $\mu$ is absolutely continuous with respect to Lebesgue measure (cf.\ Assumption A\ref{contassumption}), then there exists a unique optimal transport map $T$ between $\mu$ and $\nu$, which is cyclically monotone, i.e.\
\begin{equation*}
(x-y)\cdot(T(x)-T(y))\geq0\quad\text{for all $x,y\in\R^n$.}
\end{equation*}
\end{theorem}
Note that cyclic monotonicity coincides with the usual notion of monotonicity in the case $n=1$.

\begin{definition}[Displacement interpolation, cf.\ Remark 2.13 in Ambrosio and Gigli (2013)]\label{displacement}
Let $\mu, \nu$ be two probability measures on $\R^n$ that admit a unique optimal transport map $T$. The \emph{displacement interpolation} between $\mu$ and $\nu$ is a one-parameter family $[0,1]\ni\theta\mapsto\mu^\theta$ defined as
\begin{equation*}
\mu^\theta = \mu\circ (T^\theta)^{-1},
\end{equation*}
where the map $T^\theta:\R^n\to\R^n$ is given by $T^\theta=(1-\theta)Id+\theta T$.
\end{definition}
Clearly, $\mu^0=\mu$ and $\mu^1=\nu$, and it is known that the curve $\{\mu^\theta\}_{\theta\in[0,1]}$ is the unique geodesic with respect to the Wasserstein distance connecting $\mu$ and $\nu$.  

We will call $\theta$ the \emph{group fairness parameter}. It is a fundamental modelling parameter that allows to choose any worldview between WYSIWYG and WAE.

As a final tool, we need the notion of \emph{barycenter} of a family $\{\mu_k\}_{1,\ldots,G}$ with corresponding weights $\{w_k\}_{1,\ldots,G}$:

\begin{theorem}[Barycenter in Wasserstein space (Agueh and Carlier (2011))]\label{bary} 
Let $\{\mu_k\}_{1,\ldots,G}$ be a family of probability measures satisfying Assumption A\ref{contassumption}, and let $\{w_k\}_{k=1,\ldots,G}$ be positive weights with $\sum_{k=1}^Gw_k=1$. Then there exists a unique probability measure $\nu$ on $\R^n$ that minimizes the functional
\begin{equation*}
\nu\mapsto \sum_{k=1}^Gw_kW^2_2(\mu_k,\nu).
\end{equation*}  
This measure is called the \emph{barycenter} of $\{\mu_k\}_{1,\ldots,G}$ with weights $\{w_k\}_{1,\ldots,G}$.
\end{theorem}

The term ``barycenter" is motivated by analogy with the Euclidean case, where the center of mass of points $x_k$ with weights $w_k$ is precisely the least square approximation of these weighted points, i.e.\ the minimiser of the expression $\sum_k w_k|x-x_k|^2$. In the special case of two groups with score distributions $\mu_1, \mu_2$ and weights $w_1, w_2$, respectively, the barycenter is given by the displacement interpolation\footnote{Recall the map $T^\theta$ is defined in Definition~\ref{displacement}.}:
\begin{equation*}
\nu = \mu_1\circ (T^{w_2})^{-1},
\end{equation*} 
where $T$ is the optimal transport map between $\mu_1$ and $\mu_2$ and $T^{w_1}$ is defined as in Definition \ref{displacement}. We will need some elementary {and well-known} properties of the barycenter, {which are consequences of the shift invariance of Wasserstein distance} (we assume from now on that every $\mu_k$ satisfies Assumption A\ref{contassumption}):
\begin{lemma}\label{shift}
Let $\mu$ be a probability measure on $\R^n$, then the translation of $\mu$ by a vector $z\in\R^n$ is defined as 
\begin{equation*}
\nu_z(B)=\nu(B-z)
\end{equation*}
for any measurable $B\subset\R^n$. It holds that
\begin{itemize}
\item[i)] If $\nu$ is the barycenter of measures $\{\mu_k\}$ with weights $\{w_k\}$, $k=1,\ldots, G$, then the barycenter of the translations $\{(\mu_k)_{z_k}\}$ with the same weights is given by $\nu_{\bar{z}}$, where $\bar{z}=\sum_{k=1}^Gw_kz_k$.
\item[ii)] If $T$ is the optimal transport map from $\mu$ to $\nu$, then $T+z$ is the optimal transport map from $\mu$ to $\nu_z$.
\item[iii)] The barycenter of measures with zero expectation has itself zero expectation.
\end{itemize}
\end{lemma}
\begin{proof}
We denote the expectation of a probability measure by $\operatorname{E}$.
i) Suppose this were not the case, then there would exist another probability measure $\lambda$ on $\R^n$ such that 
\begin{equation*}
 \sum_{k=1}^Gw_kW^2_2((\mu_k)_{z_k},\lambda)< \sum_{k=1}^Gw_kW^2_2((\mu_k)_{z_k},\nu_{\bar{z}}).
\end{equation*}
Let $T_k$ be the optimal transport map from $\mu_k$ to $\nu$ (whose existence is guaranteed by Proposition 3.8 in Agueh and Carlier 2011), then the map
\begin{equation*}
T_k^*: x\mapsto T_k(x-z_k)+\bar{z}
\end{equation*}
is a transport map from $(\mu_k)_{z_k}$ to $\nu_{\bar{z}}$, and therefore
\begin{equation}\label{shift1}
 \sum_{k=1}^Gw_kW^2_2((\mu_k)_{z_k},\lambda)< \sum_{k=1}^Gw_k\int_{\R^n}|T_k^*(x)-x|^2d(\mu_k)_{z_k}.
\end{equation}
We compute
\begin{equation}\label{shift2}
\begin{aligned}
\sum_{k=1}^Gw_k\int_{\R^n}|T_k^*(x)-x|^2d(\mu_k)_{z_k}&=\sum_{k=1}^Gw_k\int_{\R^n}|T_k(x-z_k)-(x-z_k)+\bar{z}-z_k|^2d(\mu_k)_{z_k}\\
&=\sum_{k=1}^Gw_k\int_{\R^n}|T_k(x)-x+\bar{z}-z_k|^2d\mu_k\\
=\sum_{k=1}^Gw_k\int_{\R^n}|T_k(x)-x|^2d\mu_k&+\sum_{k=1}^Gw_k|\bar{z}-z_k|^2+\sum_{k=1}^Gw_k(\bar{z}-z_k)\cdot(\operatorname{E}[\nu]-\operatorname{E}[\mu_k]).\\
\end{aligned}
\end{equation}
On the other hand, let $S_k^*$ be the optimal transport map from $(\mu_k)_{z_k}$ to $\lambda$. As translations are invertible, it is possible to write 
\begin{equation*}
S_k^*: x\mapsto S_k(x-z_k)+\bar{z}
\end{equation*}
for some transport map $S_k$ from $\mu_k$ to $\lambda_{-\bar{z}}$. Then an analogous computation yields
\begin{equation}\label{shift3}
\begin{aligned}
\sum_{k=1}^Gw_k\int_{\R^n}|S_k^*(x)-x|^2d(\mu_k)_{z_k}&\\
=\sum_{k=1}^Gw_k\int_{\R^n}|S_k(x)-x|^2d\mu_k&
+\sum_{k=1}^Gw_k|\bar{z}-z_k|^2+\sum_{k=1}^Gw_k(\bar{z}-z_k)\cdot(\operatorname{E}[\nu]-\operatorname{E}[\mu_k]).\\
\end{aligned}
\end{equation}
It follows then from~\eqref{shift1}, \eqref{shift2}, and \eqref{shift3} that 
\begin{equation*}
 \sum_{k=1}^Gw_kW^2_2((\mu_k),\lambda_{-\bar{z}})< \sum_{k=1}^Gw_kW^2_2(\mu_k,\nu),
\end{equation*} 
in contradiction to $\nu$ being the barycenter.

ii) Suppose not, then there would be a transport map $S^*$ from $\mu$ to $\nu_{z}$ such that 
\begin{equation*}
\int_{\R}^n|S^*(x)-x|^2d\mu(x)<\int_{\R}^n|T(x)+z-x|^2d\mu(x).
\end{equation*}
As before, we may write $S^*(x)=S(x)+z$ for some transport map $S$ from $\mu$ to $\nu$, and then it holds that
\begin{equation*}
\int_{\R}^n|S^*(x)-x|^2d\mu(x)=\int_{\R}^n|S(x)-x|^2d\mu(x)+|z|^2+z\cdot(\operatorname{E}[\nu]-\operatorname{E}[\mu])
\end{equation*}
and analogously
\begin{equation*}
\int_{\R}^n|T(x)+z-x|^2d\mu(x)=\int_{\R}^n|T(x)-x|^2d\mu(x)+|z|^2+z\cdot(\operatorname{E}[\nu]-\operatorname{E}[\mu]).
\end{equation*}
It follows that 
\begin{equation*}
\int_{\R}^n|S(x)-x|^2d\mu(x)<\int_{\R}^n|T(x)-x|^2d\mu(x),
\end{equation*}
in contradiction to the the optimality of $T$.

iii) By ii), we have for any $z\in\R^n$ and any probability measure $\nu$ that
\begin{equation*}
\sum_{k=1}^GW_2^2(\mu_k,\nu_z)=\sum_{k=1}^GW_2^2(\mu_k,\nu)+|z|^2+2z\cdot\operatorname{E}[\nu],
\end{equation*}
where we used $\operatorname{E}[\mu_k]=0$. Choosing $z=-\operatorname{E}[\nu]$, we discover
\begin{equation*}
\sum_{k=1}^GW_2^2(\mu_k,\nu_{-\operatorname{E}[\nu]})=\sum_{k=1}^GW_2^2(\mu_k,\nu)-|\operatorname{E}[\nu]|^2,
\end{equation*}
so that any minimizer must indeed have the property $\operatorname{E}[\nu]=0$.
\end{proof}

We are finally ready to define our continuous fairness algorithm and prove its optimality.

\subsection*{Continuous Fairness Algorithm (CFA$\theta$)} Given $\mu_k$ satisfying Assumption A\ref{contassumption} and $w_k$ the corresponding weights, let $\nu$ be the barycentre. Then choose $\nu^k$ to be the displacement interpolation $\mu_k^\theta$ between $\mu_k$ and $\nu$, and $T_k$ to be the optimal transport map between $\mu_k$ and $\nu_k$ (cf.\ Theorem~\ref{brenier}). \\

For $\theta=0$ and $\theta=1$, we reobtain the algorithms described above for WYSIWYG and WAE, respectively. As a generalization, we may also pick different values of $\theta$ for different $k$\footnote{ This may be significant in order to handle intersectionality, as explained in the introduction.}. In this case, we would choose $\nu^k$ as the displacement interpolation $\mu_k^{\theta_k}$ between $\mu_k$ and $\nu$.

The question remains how the algorithm (CFA$\theta$) performs with regard to our fairness requirements (\ref{individual})--(\ref{utility}). Before we answer these questions, we have to formulate these requirements in a more precise way.

Consider the fairness criteria in order:
\begin{enumerate}
\item As discussed above, individual fairness will typically be in conflict with group fairness. As a consequence, it only makes sense to consider individual fairness \emph{within groups} (see also Dwork et al. 2012, p. 220).  Individual fairness is usually formulated in terms of continuity or Lipschitz continuity (Dwork et al. 2012; Friedler et al. 2016). In our notation, this would mean that there exists a Lipschitz constant $L>0$ such that
\begin{equation}\label{lip}
|T_k(x)-T_k(y)|\leq L|x-y|\quad\text{for all $x,y\in[0,1]$ and $k=1,\ldots,G$.}
\end{equation} 
However, this will be impossible to satisfy in general,\footnote{ The main obstruction to continuity of transport maps is the possible non-connectedness of the support of the target measure (cf.\ Theorem 1.27 in Ambrosio and Gigli 2013). This is the case if, within one group, there are further subgroups with very different score statistics.} although under certain assumptions, like the convexity of the support of the target measure, it can be guaranteed (Theorem 1.27 in Ambrosio and Gigli 2013). Instead, we aim to characterize individual fairness \emph{in a least square sense}: To this end, we consider the quantity
\begin{equation}\label{error}
\sum_{k=1}^Gw_k\int_{\R^n}\int_{\R^n}\frac12|T_k(x)-T_k(y)|^2d\mu_k(x)d\mu_k(y).
\end{equation}
Clearly, it is non-negative, and is zero if and only if $T_k$ takes a constant value $\mu_k$-almost everywhere, for all $k=1,\ldots, G$. This last observation implies that it is unreasonable to minimize this quantity, as the minimizing configuration would assign the same fair score to every individual. Not only would such a classification be completely useless for the decision maker, but it is also questionable if it would be ``fair" in any meaningful sense of the word. Indeed, being treated similarly as a much less qualified member of one's group is hardly any fairer than being treated completely differently than a similarly qualified one. This indicates that in~\eqref{lip}, the Lipschitz modulus $L$ should not be \emph{minimal}, but \emph{as close as possible to one}. In our least-square framework, therefore, instead of minimizing~\eqref{error}, we should rather minimize
\begin{equation}\label{indfair}
E_{ind}:=\sum_{k=1}^Gw_k\int_{\R^n}\int_{\R^n}\frac12\left|(T_k(x)-T_k(y))-(x-y)\right|^2d\mu_k(x)d\mu_k(y),
\end{equation}
where we call $E_{ind}$ the \emph{individual fairness error}.

Note that this definition of individual fairness, in contrast to Dwork et al. (2012) and other papers in this vein, does not presuppose a similarity metric. This seems justified for two reasons. First, the establishment of a similarity metric, providing an ``objective'' assessment of the candidates with respect to the desired trait (ground truth), has been recognized as extraordinarily difficult and, in fact, as the main shortcoming Dwork et al.'s framework in the literature (Zemel et al., 2013; Fish et al., 2016, p. 146; Friedler et al., 2016, p. 3; Gillen et al., 2018, p. 1 et seq.; Chouldechova and Roth, 2018, p. 4), but also in Dwork et al. (2012), p. 214 et seq.\ itself. Indeed, it seems unclear how such a metric should be established in a way that is simultaneously less prone to bias than the raw score and scalable to hundreds or thousands of applicants. Quite tellingly, one example of a similarity metric Dwork et al. (2012), p. 215, give is one based on credit worthiness. However, credit worthiness scores (e.g., FICO scores) are precisely a type of raw scores we would work with, too. Therefore, in practice, the difference between raw scores and a similarity metric disappears given real-world limitations on the establishment of the latter (see also Joseph et al., 2016).

Second, if one had access to an objective similarity metric, the whole scoring procedure would be quite superfluous as one could use this metric directly to make selection decisions. Hence, we adopt a realist perspective in positing that access to such a metric will generally be impossible, and that we will have to content ourselves, for purely pragmatic reasons, with the raw scores arising from the (machine learning) regression procedure. This seems sensible as long as these raw scores are somewhat close to ground truth, even though they may be tainted with bias. If, however, the raw scores are so disjunct from reality (e.g., because of bias) that they provide no meaningful information to guide selection procedures, then rather than salvaging the model through algorithmic fairness, one should discard it altogether. In sum, our mapping exercise takes the raw scores as a realistically accessible baseline and therefore dispenses of the need for a similarity metric. Notwithstanding, our model can obviously be paired with fairness procedures seeking to ensure the establishment of rather objective raw scores (see, e.g., Zemel et al., 2013). Note, however, that even ``objective'', ``factually correct'' raw scores may lead to discrimination (in the sense of violation of group fairness) if the desired trait is unevenly distributed across protected groups, and that this may still necessitate fairness interventions of the kind we envisage (cf. Gilles et al., 2018, p. 1).

\item The requirement of statistical parity imposes that the fair representations conceal any information on group membership. This means that the events ``an individual belongs to group $k$" and ``an individual has score $s$" should be stochastically independent. More precisely, we demand for each $k\in 1,\ldots,G$ and every measurable $B\subset\R^n$
\begin{equation}\label{indep}
w_k\nu_k(B)=w_k\sum_{l=1}^Gw_l\nu_l(B).
\end{equation} 
Let now $v_1,\ldots,v_G$ be any other set of positive numbers with $\sum_{k=1}^Gv_k=1$. Then multiplying~\eqref{indep} with $v_k$ and summing over $k$ yields
\begin{equation*}
\sum_{k=1}^G(v_k-w_k)\nu_k=0.
\end{equation*}
For any measurable $B\subset\R^n$, consider the vector $\bar{\nu}:=(\nu_1(B),\ldots,\nu_G(B))$. Since the choice of the $v_k$ was arbitrary, the vector $(v_1-w_1,\ldots,v_G-w_G)$ can run through a relatively open subset of the linear subspace $\{x:\sum_{k=1}^Gx_k=0\}\subset \R^G$, and hence $\bar{\nu}$ is orthogonal to this subspace. But the orthogonal complement of this subspace is the span of the vector $(1,\ldots,1)$, so we conclude that all components of $\bar{\nu}$ are equal. Since $B$ was chosen arbitrarily, it follows that all the measures $\nu_k$, $k=1,\ldots,G$, are identical.

\item Again, the monotonicity requirement can only be fulfilled within groups. As $\R^n$  admits no linear ordering for $n\geq2$, we restrict ourselves to the case $n=1$. Monotonicity can be formulated as 
\begin{equation*}
T_k(x)\leq T_k(y)\quad\text{whenever $x\leq y$},
\end{equation*}
for every $k=1,\ldots, G$.
\item Finally, we assume that the decision maker utility is given as
\begin{equation*}
U=-\frac12\sum_{k=1}^Gw_k\int_{\R^n}|T_k(x)-x|^2dx,
\end{equation*}
where the negative sign reflects that the utility is higher if the squared difference of the raw score and the fair score is smaller. Note that decision maker utility is maximized when $T_k=id$, i.e.\ when the raw score and the fair score are identical. 
\end{enumerate}

Recall the situation $\theta=1$, which reflects the WAE worldview. Then all raw distributions are mapped to the same fair distribution $\nu$. We have argued in (2) above that this is necessary (and sufficient) for perfect group fairness in the sense of statistical partity. We have the following optimality result:
\begin{theorem}\label{optimalitythm}
Let $\theta=1$. Among all possible choices of the target map $\nu$ and of transport maps $T_k$ from $\mu_k$ to $\nu$, the ones specified in CFA$\theta$ maximize decision maker utility, minimize the individual fairness functional \eqref{indfair}, and, in the case $n=1$, ensure monotonicity within groups. 
\end{theorem}
\begin{proof}
By definition, the barycenter $\nu$ of the $\mu_k$ with weights $w_k$ minimizes the functional 
\begin{equation*}
\nu\mapsto\sum_{k=1}^Gw_kW^2_2(\mu_k,\nu)
\end{equation*}
among all probability measures on $\R^n$ (cf.\ Theorem \ref{bary}). On the other hand, by \eqref{cost} and \eqref{wasser}, this functional is equal to the negative utility
\begin{equation*}
\sum_{k=1}^Gw_k\int_{\R^n}|T_k(x)-x|^2d\mu_k(x)
\end{equation*}
if and only if $T_k$ is the unique optimal transportation map from $\mu_k$ to $\nu$. This already establishes maximal utility for our choice of $\nu$, $T_k$. Note also that this choice is the unique one with this property. 

By Brenier's Theorem \ref{brenier}, the optimal transport maps $T_k$ are cyclically monotone, which implies monotonicity in the one-dimensional setting. So it remains to show that \eqref{indfair} is also minimized by our choices of $\nu$, $T_k$. Notice first that $E_{ind}$ is invariant under translations of the measures $\mu_k$ and $\nu$. Thus, given $\{\mu_k\}$ and $\{w_k\}$ for $k=1,\ldots,G$, the minimal value of $E_{ind}$ is attained replacing $\mu_k$ by $\bar{\mu}_k$, where $\bar{\mu}_k$ is the translated version of $\mu_k$ with zero expectation, and minimizing only with among measures $\bar{\nu}$ with zero expectation.

Denoting by E the expectation of a measure, we then compute for the individual fairness error of $\bar{\mu}_k$, $w_k$, $\bar{\nu}$, and $\bar{T}_k$ any transport maps from $\bar{\mu}_k$ to $\bar{\nu}$:

\begin{equation*}
\begin{aligned}
E_{ind}&=\sum_{k=1}^Gw_k\int_{\R^n}\int_{\R^n}\frac12\left|(\bar{T}_k(x)-\bar{T}_k(y))-(x-y)\right|^2d\bar{\mu}_k(x)d\bar{\mu}_k(y)\\
&=\sum_{k=1}^Gw_k\left[\int_{\R^n}|\bar{T}_k(x)-x|^2d\mu_k(x)-\int_{\R^n}\int_{\R^n}(\bar{T}_k(x)-x)\cdot(\bar{T}_k(y)-y)d\bar{\mu}_k(x)d\bar{\mu}_k(y)\right]\\
&=-2U-\sum_{k=1}^Gw_k|\operatorname{E}[\bar{\nu}]-\operatorname{E}[\bar{\mu}_k]|^2=-2U.\\
\end{aligned}
\end{equation*} 
But we have already seen that the negative utility is minimized by choosing $\bar{\nu}$ as the barycenter of the $\bar{\mu}_k$ with weights $w_k$, and $\bar{T}_k$ the optimal transport maps from $\bar{\mu}_k$ to $\bar{\nu}$. This choice is consistent with our requirement that $\operatorname{E}[\bar{\nu}]=0$, since the barycenter of measures with expectation zero has itself expectation zero by virtue of Lemma~\ref{shift}iii).

Translating back again and recalling the invariance of $E_{ind}$ under such translations, we see that the translated measure $\nu:=\bar{\nu}_{\bar{z}}$, where $\bar{z}:=\sum_{k=1}^Gw_k\operatorname{E}[\mu_k]$, and the translated transport maps $T_k:=\bar{T}_k+\operatorname{E}[\mu_k]$ minimizes $E_{ind}$ for our original measures $\mu_k$ among all possible choices of $\nu$ and $T_k$. The statement of the Theorem then follows by observing that $\nu$ thus chosen is precisely the barycenter of the $\mu_k$ with weights $w_k$ by Lemma~\ref{shift}i), and the $T_k$ are the optimal transport maps from $\mu_k$ to $\nu$ by Lemma~\ref{shift}ii).

\end{proof}

\section{Examples}

In this section, we consider three examples that highlight different features of our fairness framework: college admissions; credit decisions; and insurance contracts.

\subsection{College Admissions}

We start by looking at university/college admissions decisions, a common example in the literature on algorithmic fairness (see Friedler et al.\ 2016). As noted, our framework performs particularly well when groups are large; hence, it is less suitable for small colleges, but of relevance, for example, in centralized national or other large admissions procedures. As explained in greater detail in the model, our approach consists in transforming a raw score that individuals receive into a fair score that fulfills certain fairness requirements while simultaneously optimizing the utility of the decision maker under the fairness constraints. The raw score that applicants receive will, in the case of colleges admissions, be some score that aims to predict college performance based on a number of input data, such as admissions tests; high school GPA; previous job or educational experience etc. If we care about group fairness, we will be interested in how the individual raw scores are distributed in different groups, and to what extent our Continuous Fairness Algorithm CFA$\theta$ will remedy inequality between these distributions.

For the sake of analytical clarity, let us consider a drastically simplified example. We assume that there are only two applicant groups of interest (for example, two different ethnic or gender groups), each consisting of an equal, large number of members. The college will admit the top-ranked 50 percent of applicants. Individuals receive raw scores that run continuously from 0 (worst) to 1 (best). Raw scores in Group A are distributed according to a compactly supported probability density function on the interval from 0 to 0.5 (e.g.\ a normal distribution centered at 0.25, with tails cut off at 0 and 0.5); similarly, raw scores in Group B are distributed according to a density function of the same shape from 0.5 to 1, see Figure 1. Thus, even the worst member of Group B scores better than the best member of Group A. In this extreme case, if the college was to base its admissions decision on the raw scores, it would accept all members of Group B and no members of Group A. While maintaining individual fairness within groups and also between members of different groups, this outcome violates group fairness to an extreme extent. 

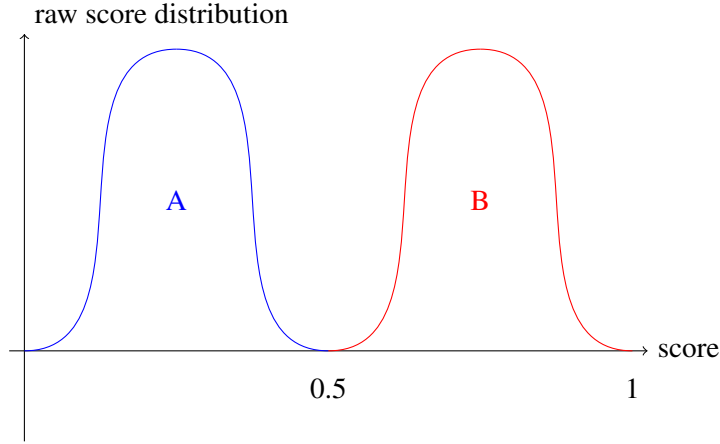
\begin{figure}
\begin{center}
\caption{Distribution of raw scores for two groups}
\vspace{0.4cm}
 \begin{tikzpicture}[domain=0:8] 
    \draw[->] (-0.2,0) -- (8.2,0) node[right] {score}; 
    \draw[->] (0,-1.2) -- (0,4.2) node[above right] {raw score distribution};
    \draw[color=blue] (0,0) to[out=0,in=180] (2,4);
    \draw[color=blue] (2,4) to[out=0,in=180] (4,0);
\draw(2,2) node[color=blue]{A};
  \draw[color=red] (4,0) to[out=0,in=180] (6,4);
    \draw[color=red] (6,4) to[out=0,in=180] (8,0);
\draw(6,2) node[color=red]{B};
\draw (4,-0.5) node{0.5};
\draw (8,-0.5) node{1};
  \end{tikzpicture}
	\end{center}
\end{figure}

To apply our fairness framework, we first have to find the intermediate distribution that constitutes the barycenter. This is the distribution that minimizes the sum of the squares of the distances to both raw group distributions (in the Wasserstein metric). In our case, it is easy to see that it is the distribution of the same shape as the raw distributions which runs from 0.25 to 0.75 and which peaks at a score of 0.50. Note, however, that the characterization of the barycenter will not be as easy in more general circumstances, when the two (or more) distributions are no longer merely translates of each other.

\begin{figure}
\begin{center}
\caption{The barycenter distribution}
\vspace{0.4cm}
 \begin{tikzpicture}[domain=0:8] 
    \draw[->] (-0.2,0) -- (8.2,0) node[right] {score}; 
    \draw[->] (0,-1.2) -- (0,4.2) ;
     \draw[color=violet] (2,0) to[out=0,in=180] (4,4);
    \draw[color=violet] (4,4) to[out=0,in=180] (6,0);
\draw(4,2) node[color=violet]{barycenter};
\draw (4,-0.5) node{0.5};
\draw (8,-0.5) node{1};
  \end{tikzpicture}
	\end{center}
\end{figure}
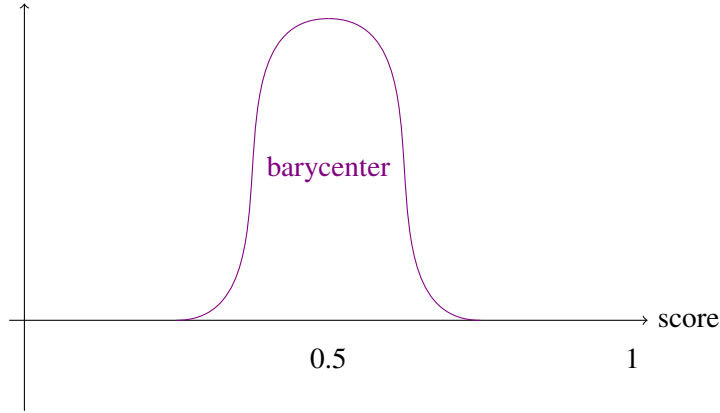
By choosing $\theta$, we can now vary to what extent we would like to transport the raw distributions of each group toward this intermediate distribution. If we choose $\theta$ = 1, our CFA$\theta$ maps the raw distributions of both groups completely onto the intermediate distribution (cf.\ Figure 2). If we present this novel distribution to the college, and the college applies its decision rule of accepting the top-ranking 50 percent, it will accept all those with fair scores of 0.5 or more.\footnote{ If the group size is odd, the college will have to randomize its admission decision for the lowest ranking pair of individuals at or above score 0.5. Of course this issue is not visible in our continuum framework.} In this case, group fairness is fully safeguarded as it is impossible to deduce group membership from the fair score: the fair distribution is exactly the same for both groups. This is equivalent with the observation that an equal percentage of applicants is expected from each group (statistical parity). It can be understood as a reconstruction of the equal distribution of scores among groups in the construct space (WAE) as defined by Friedler et al. (2016). Within groups, our use of an optimal transport map from the raw to the fair distribution guarantees monotonicity: all those that received a higher raw score than other group members receive a higher, or equal, fair score. Furthermore, the optimality constraint on the transport ensures that decision maker utility is maximized: everyone gets a fair score that is as close as possible, in the least square sense, to their raw score. However, monotonicity is violated between groups: members of the lower-ranking half of Group B each had a higher raw score than the members of the upper-ranking half of Group A, but the ranking is inverted for the fair scores. By implication, individual fairness between members of different groups is violated; for example, the top-ranking member of Group A had a similar raw score to the lowest-ranking member of Group B. However, they are now at opposite ends of the spectrum of fair scores. This is exactly the price one has to pay for group fairness, and the reason why affirmative action continues to be so contentious (see also the discussion in the subsection on legal and policy implications).

This does not exhaust the possibilities our algorithm offers, however. Rather, it admits for any degree of approximation between the raw scores, ranging from zero approximation ($\theta$ = 0: fair scores = raw scores) to the full approximation just discussed ($\theta$ = 1: fair scores of both groups distributed along the barycenter). For example, if $\theta$ = 1/2, each of the distributions is transported ``halfway" toward the barycentre distribution, as shown in Figure 3. 

\begin{figure}
\centering
\caption{Fair representations for $\theta=\frac12$}
\vspace{0.4cm}
 \begin{tikzpicture}[domain=0:8] 
    \draw[->] (-0.2,0) -- (8.2,0) node[right] {fair score}; 
    \draw[->] (0,-1.2) -- (0,4.2) node[above right] {density};
    \draw[color=blue] (1,0) to[out=0,in=180] (3,4);
    \draw[color=blue] (3,4) to[out=0,in=180] (5,0);
\draw(3,2) node[color=blue]{A};
  \draw[color=red] (3,0) to[out=0,in=180] (5,4);
    \draw[color=red] (5,4) to[out=0,in=180] (7,0);
\draw(5,2) node[color=red]{B};
\draw (4,-0.5) node{0.5};
\draw (8,-0.5) node{1};
  \end{tikzpicture}
\end{figure}
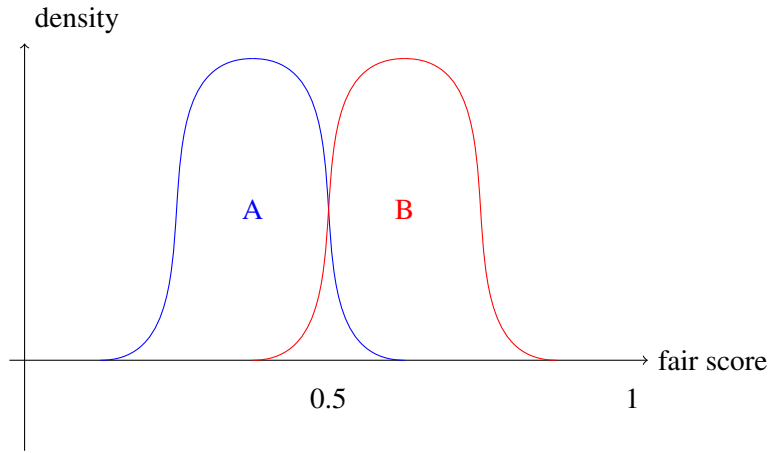

As noted in the introduction, $\theta$ can also be chosen so that the ``80 percent rule" is fulfilled. This would mean that the probability of a member of Group A of being admitted is at least 80 percent of the respective probability of a member of Group B. Such a result can also be achieved by imposing strict quotas. Choosing a $\theta$ value instead of a certain quota, however, affords the advantage that it only imposes a certain degree of approximation of the raw distributions; it does not impose a fixed number of applicants admitted from either group. This is not only an advantage from a legal perspective, see Section~\ref{sec:legal}. It also implies that if in one year Group B has particularly strong candidates (but still worse than A), they are not penalized. In every case, the college must base its admissions decision on the fair distribution that results from the choice of a specific $\theta$. As Section~\ref{sec:legal} will highlight in greater detail, it is thus crucial to choose $\theta$ wisely. 

We note in passing that the preceding discussion easily carries over to private companies or public agencies making hiring decisions; this is another potentially highly relevant field for algorithmic fairness since cases of algorithmic bias have already been documented in this realm, see, e.g. Lowry and Macpherson (1988); Reuters (2018). 

\subsection{Credit Applications}

Our second major example concerns a decision maker such as a bank that decides on applications for credit. As noted in the introduction, this is an area where algorithmic decision making has already become firmly rooted, with algorithmically determined credit rating scores such as the FICO score in the US and the SCHUFA score in Germany being in wide use. A number of Fintech startups are basing their loan decisions entirely on algorithmic models (Hurley and Adebayo 2016).\footnote{ See, e.g., https://www.kreditech.com/.} Scholars, however, have criticized the credit rating system as opaque and biased (see Pasquale 2015; Rothmann et al. 2014; Hurley and Adebayo 2016). Let us therefore consider an example in which individual credit applicants are assigned a credit rating raw score based on a number of factors such as their financial history, available collateral etc. The decision rule for the institution consists in establishing a cut-off threshold at a certain score below which loan applications will be rejected. Above this score, applicants receive loan offers. The better their scores, the better the loan conditions they are offered.

Again, the CFA$\theta$ forces the institution to base its decision on a modified, fair score instead of the raw score. Just like in the previous example, the choice of a concrete $\theta$ is paramount to determine to what extent different distributions in the various applicant groups should be approximated to the barycenter distribution. Importantly, let us assume that we have two relevant, binary categories (for example, ethnicity: black and white; gender: male and female) that create four different groups. If, for example, there is evidence of particular raw score bias against blacks and against male applicants separately, but moreover of even greater bias against black male applicants,\footnote{ This pattern is basically found, for example, in the empirical study by Ayres and Siegelman (1995), p. 309 et seq., for offers made by car dealers to members of the respective groups. In other settings, of course, women may be the group discriminated against.} we can choose a higher $\theta$ for the transformation of the scores belonging to the latter group. This pushes these particularly biased raw scores closer to the intermediate distribution than the less biased raw scores of the black and white female, and white male, applicants. Hence, our framework accommodates different degrees of bias arising from intersectionality.

\subsection{Insurance Contracts}

A final example stems from insurance contracts. Among standard consumer contracts, insurance contracts are probably the ones that are based on the most refined statistical models. If they were not, insurance companies would be out of business fairly soon. Increasingly, insurance companies are offering personalized insurance contracts whose conditions, including premiums, are determined by a complex set of factors (see, e.g., Guan 2016). Risk scores are calculated for a large number of individuals (see Chen et al. 2017, p. 201) so that our framework generally applies. Let us therefore imagine that an insurance company seeks to distill a risk score from the input data. In the case of car insurance, this may include accident history; driving style data; car model; age group etc. 

Again, there may be significant differences in the distribution of the raw risk scores among relevant groups. What makes the case of insurance so special, at least in the EU, is the fact that the Court of Justice of the European Union, in its landmark Test-Achats decision in 2011, ruled that gender may not be a determining factor in insurance pricing.\footnote{ CJEU case C-236/09 \emph{Test-Achat} ECLI:EU:C:2011:100, para. 28--34; on this, see Tobler (2011).} While we cannot offer a comprehensive analysis of the legal implications of this judgment (see Hacker 2018, particularly p. 1166 et seq. for more detail), it is clear that, for example in the case of car insurance, insurers may comply with the ruling by not offering women cheaper premiums than men despite statistical evidence that the former are safer drivers (see Reed et al. 2016, p. 3). Therefore, insurance companies should have an intrinsic interest in transforming raw risk scores such that gender differences are mitigated or erased from the data. Our algorithm provides for a straightforward way of achieving this: we only have to take the distributions for the male and the female group, respectively; calculate the barycenter; and fully transport the scores onto that distribution by setting $\theta = 1$. Depending on how much leeway the law leaves to insurance companies,\footnote{ See, e.g., the guidelines by the European Commission stating that "[t]he use of risk factors which might be correlated with gender [...] remains possible, as long as they are true risk factors in their own right" (European Commission, Guidelines on the application of Council Directive 2004/113/EC to insurance, in the light of the judgment of the Court of Justice of the European Union in Case C-236/09 (Test-Achats), OJ 2012 C 11/1, para. 17). Thus, the legal admissibility of correlated factors crucially depends on whether these factors can plausibly be related to the risks covered by the insurance. In machine learning contexts, where specific factors may not always be reconstructable from the output (particularly in deep neural networks), insurers can ``play it safe" by approximating male and female scores.} they (or an impartial agent) may merely approximate male and female score distributions to one another, instead of making them identical. Here again, it may be of interest to fulfill the ``80 percent rule".

\section{Relation to Other Work}\label{relation}

The field of algorithmic fairness proliferates and is generating a staggering number of definitions of algorithmic fairness. One branch of the field is concerned with \emph{detecting} discrimination in algorithmic decision making (see \v{Z}liobait\.e 2017). The second branch, which concerns us here, seeks to \emph{reduce} discrimination by aligning decision processes with formal definitions of fairness. Without laying claim to completeness, we see definitions as falling into four different categories that can be roughly attributed to different stages in the algorithmic process. First, there are data reconfiguration approaches (also called ``data massaging", in Zemel et al. 2013) that seek to transform \emph{input data} into a fairer representation (pre-processing approaches); examples are Zemel et al. (2013) and Calmon et al. (2017). A second set of approaches seeks to control the \emph{algorithmic process} from the input to the output (in-processing approaches). One subset of this group establishes constraints via a distance metric, generally to ensure individual fairness. Here, examples are Dwork et al. (2012) and Friedler et al. (2016). The distance metric, in their approaches, measures distances between individuals or groups in the observed space. Then, a constraint is introduced to ensure that the output, in what Friedler et al. (2016) call the decision space, does not significantly increase distances between groups or individuals. Another subset similarly seeks to restrain the mapping from input to output data, but by positing equal output probabilities at the group level, such as statistical parity. Examples include again Dwork et al. (2012) and Friedler et al. (2016) with their group fairness definition, but also Chouldechova (2017), Kleinberg et al. (2016), Berk et al. (2018) and Datta et al. (2015) (calibration, statistical parity). Third, another type of data reconfiguration approach transforms the \emph{output data} (post-processing approaches), in our case a raw score delivered by an algorithm, or some other decision procedure. Fourth, a final approach is based on external performance measures, i.e., metrics that analyze in how far the predictions fit real outcomes (reality check approaches); examples include Chouldechova (2017), Kleinberg et al. (2016) and Berk et al. (2018) with measures such as predictive parity, conditional use accuracy equality and error rate balance, as well as Hardt et al. (2017) with equalized odds.

As should have become clear during the discussion of our model, our model uses a post-processing approach. We note that we could apply our transformation to the raw scores calculated in the training phase of a machine learning model; in this way, it could extend from the post- to the pre-processing phase. As a post-processing method, our approach is related to the ones employed by Zemel et al. (2013), Dwork et al. (2012) and Friedler et al. (2016). However, as opposed to Dwork et al. (2012) and Friedler et al. (2016), we do not operate with a general distance metric, and do not include a general Lipschitz condition to guarantee that distances stay approximately the same in Wasserstein space (however, we can guarantee Lipschitz continuity under some constraints, see \eqref{lip}); rather, we measure individual fairness in a least square sense, and aim to ensure fairness by monotonicity, which in turn is achieved through optimal transport. Our raw score can, as discussed above, be understood as a functional equivalent of the distance metric in Dwork et al. (2012, p. 215, 224), with the difference that we \emph{generally} (although not necessarily) envision the raw score already as the result of an algorithmic process. At the group level, for $\theta$ = 1, our model enforces group fairness, understood as statistical parity, and can be understood as a reconstruction of the equal distribution of scores among groups in the construct space as defined by Friedler et al. (2016). We provide a specific measure for this equal distribution, however, by the construct of the barycenter. Furthermore, by varying $\theta$ via displacement interpolation, we are able to continuously navigate between the extreme worldviews (WAE, WYSIWYG) presented in Friedler et al. (2016), and between individual and group fairness. In this, the results of our transformation bear resemblance to Calmon et al. (2017), who, in a pre-processing and discrete setting, seek to achieve group fairness while accounting for individual fairness via a distortion constraint. To the extent that we choose statistical parity as a measure for group fairness, our model is also similar to \v{Z}liobait\.e (2015), who introduces an equal acceptance rate. Finally, in so far as we can guarantee a fulfillment of the ``80 percent rule", our approach is akin to Feldman et al. (2015) and Zafar et al. (2017). 

Another paper that introduces a post-processing constraint that allows for varying degrees of group fairness is Zehlike et al. (2017). It does not, however, operate with the notion of the barycenter. Rather, they use a cumulative distribution function to guarantee (subject to a significance parameter) that a ranking includes \emph{at least} a certain proportion of members from a protected group in its top ranks. One advantage of that measure is that it enforces only a minimum constraint: it leaves the ranking intact if (in the raw ranking) more members of the protected group are included than specified in the proportion constraint. This could be relevant particularly if the disadvantaged group has one low-scoring subgroup, but also one high-scoring subgroup that performs better than the advantaged group. Our framework achieves a similar result if we define a $\theta$ value for a selection process without fixing strict admissions quota; if the disadvantaged group scores particularly high in one year, they would end up closer to the advantaged group, and hence receive higher selection rates. In the alternative, we may split the disadvantaged group into a high- and a low-scoring one, and apply a low $\theta$ to the higher scoring one. This would preserve the high scores of this subgroup, while one may transport the low scores of the other subgroup towards the barycenter by means of a higher $\theta$. 

Furthermore, some papers have now started to apply optimal transport to algorithmic fairness. {Optimal transport has the significant advantage that it allows to manipulate not only partial aspects of the score statistics (like the mean), but the full statistics. It is thus a very natural tool to produce statistical parity. There is a considerable body of literature on fair learning of regression models, such as Calders et al.\ (2013), Fukuchi et al.\ (2015), or P\'erez-Suay et al.\ (2017); regression, however, is a somewhat rigid method, as a specific (e.g.\ linear) ansatz for the predictor is postulated a priori. As a consequence, usually only the first moments, i.e.\ the mean values, of distributions of different groups can be matched, which is a much weaker notion than full statistical parity. See however Kamishima et al.\ (2018) for a matching of the first two moments for a particular model for rating prediction. Optimal transport, on the other hand, allows for genuinely nonlinear transformations of a much more flexible kind and thereby is able to ensure statistical parity. Moreover, it does so in a provably optimal way, as demonstrated by Theorem~\ref{optimalitythm}. }

Del Barrio et al. (2019), which cites our original working paper, similarly formalizes the trade-off between information loss and group fairness. However, they work in the context of classification, restrict their algorithm to two groups, and use a random repair algorithm which is distinct from our linear displacement. Similarly, Johndrow and Lum (2019) propose a repair algorithm using optimal transport theory. However, like Feldman et al.\ (2015), they consider only the unidimensional case, while we generalize to a higher-dimensional setting.

Finally, none of the mentioned technical papers discusses the legal and policy implications of their model to any substantial extent, an issue we address in the following section.

\section{Legal and Policy Implications}
\label{sec:legal}

Because of the proliferation of fairness definitions (see above, Section~\ref{relation}), regulators seeking to implement fairness constraints face a serious selection problem. Against this background, our model offers a convenient and flexible framework to switch between different goals within one single model. In this sections, we briefly discuss three legal and policy implications: the relevance of a continuous scaling for the law; the choice of the substantive value of $\theta$; and possible procedures to implement it.

\subsection{The Relevance of $\theta$ for the Law}
There are two ways in which the possibility to gradually tune the degree of fairness via CFA$\theta$ would benefit regulators, but also companies or other organizations in complying with \emph{anti-discrimination law}.
First, most cases of anti-discrimination law turn on group fairness. Often, when unequal treatment results from the algorithm picking up differences between groups encoded in the training data set, the violation of group fairness can lead to what is called indirect discrimination under EU law and disparate impact under US law (the following discussion holds for both; see Barocas and Selbst 2016, p.~701 for US law; Hacker 2018, pp.~1152 et seq.\ for EU law). As mentioned, while there is no fixed quantitative threshold, legal scholars assume that indirect discrimination is found if the selection procedure deviates from statistical parity so that a member of the disfavored group has a less than 80\% (US) or 75\% (EU) chance, vis-\`a-vis a member of the favored group, of being positively selected (EEOC 2015, Section 4 D.; Barocas and Selbst 2016, p.~701; Hacker 2018, p.~1153). Already at this stage, it is therefore important for companies to be able to gradually tune the degree of group fairness. As there is no bright line quantitative test, the closer they approximate statistical parity, the more likely they are to escape the verdict of indirect discrimination in the first place.

More importantly, however, even if indirect discrimination is found, it is not illegal per se.\footnote{ This becomes apparent in the very definition of indirect discrimination, for example in Art. 2(b) of Directive 2004/113/EC on sex [i.e., gender] discrimination: ''where an apparently neutral [...] practice would put persons of one sex at a particular disadvantage compared with persons of the other sex, unless that [...] practice is objectively justified by a legitimate aim and the means of achieving that aim are appropriate and necessary.''} Rather, it can be justified through a proportionality test: the legitimate interests of the organization using the specific selection procedure (e.g., predictive accuracy with respect to the predicted trait) must outweigh the interests of the disfavored individuals or groups (Tobler 2005, pp. 241 et seq.). In this balancing exercise, the degree of differential treatment, i.e., the degree of the violation of group fairness, is an important parameter: for example, discriminatory practices may not go beyond what is absolutely necessary to reach the legitimate, competing goals of the decision maker (Tobler 2005, p.~242; Sullivan 2005, p. 963 et seq.). Clearly, therefore, the more the procedure discriminates against a certain group, the more difficult it is to justify (Hacker 2018, p.~1164). This shows that, while in the end a selection procedure will be either legal or illegal with respect to anti-discrimination law, the inner workings of this body of law go beyond this binary choice and crucially depend on the degree of discrimination (Selbst 2017, p. 165 et seq.), in other words: the degree of fairness. The algorithm proposed here therefore allows organizations to trade off (i) the degree to which they want to correct a machine-learned result with (ii) the risk of being found liable for violating anti-discrimination law, while simultaneously maximizing the utility of the decision maker at every point of that trade-off. 

Second, however, the law also imposes \emph{constraints on correcting} the results of selection procedures. If re-ranking is achieved by taking protected characteristics into account, the legal limits of positive action under EU law, or affirmative action in US law, need to be heeded, as discussed below. The basis for the societal debate around affirmative action, and for the important body of jurisprudence on it, is precisely the fact that every re-ranking in favor of the members of one protected group, while enhancing group fairness, may simultaneously reverse-discriminate against members of the other group, encroaching on individual fairness between different applicants (Moses 2016; Robinson 2016). If re-ranking is based, e.g., on ethnicity, otherwise similarly situated individuals will be treated differently. In this case, norms of individual equality (equal protection clauses) conflict with provisions of group equality (anti-discrimination law). This is precisely why affirmative action is so contested, and why litigation over it continues. In affirmative action, a complex body of case law has developed both under US (overview in Robinson 2016; Poueymirou 2017; Nainbandian 2000) and EU law (Craig and de B\'urca 2011, p. 909 et seqq.). Importantly, both the US Supreme Court and the CJEU stress that re-ranking remains possible during the ``test-design stage'' of a selection procedure, i.e., before selection results have been allocated to the respective candidates (Bent 2019, p.~35 et seqq.; Kroll et al. 2017, p.~695; Barocas and Selbst 2016, p.~725; Hacker, 2018, p.~1181). 

This leaves ample room for fairness strategies during the development of the machine learning model. As mentioned, our model can be fruitfully applied at the test-design stage (by applying it to the scores of the training data set). However, if the discriminatory effect of the selection algorithm is only discovered once the candidates have already been screened, and raw scores have been assigned to all of them (the selection stage), options for correction are arguably more limited under the law (Waddington and Bell 2001, p.~600). In this case, a delicate balance needs to be struck between individual and group equality, in other words: individual and group fairness (see the discussion in Kroll et al. 2018, pp.~694 et seq.; Kim 2017, pp.~199 et seqq.; Hacker 2018, pp.~1180 et seq.). The choice between them therefore reproduces not only the dilemmas of affirmative action, but also the controversies between meritocratic and outcome-egalitarian conceptions of social justice. For EU law, the CJEU has ruled that at the selection stage, an automatic and unconditional preference of one candidate over another because of protected criteria is incompatible anti-discrimination law.\footnote{ CJEU, case C-450/93, Kalanke, EU:C:1995:322, para 22.}  Rather, it is necessary to take the concrete circumstances of the case, and the qualifications of the candidates, into account when making re-ranking decisions.\footnote{ CJEU, case C-409/95, Marschall, EU:C:1997:533, para 33.} Hence, some degree of human oversight of the re-ranking is necessary -- the results of the fair algorithm need to be checked by a mechanism ensuring that re-ranking based on protected criteria is not automatic. Similar constraints of holistic review are in place under US law (Robinson 2016, pp. 192-194; cf. also Malamud, 2015 p.~14). Hence, fixed minority quota at the selection stage would fail under both EU and US law.\footnote{ CJEU, case C-450/93, Kalanke, EU:C:1995:322, para 22; US Supreme Court, Fisher II, 136 S. Ct., p. 2210.}

As mentioned, the algorithm proposed here does not enforce strict quotas, even under a fixed $\theta$ value. Rather, it precisely facilitates the legally required trade-off and fine-tuning to the concrete case. It does not uniformly enforce one set of quotas or, more generally, one specific re-ranking outcome; instead, $\theta$ can be chosen consciously with respect to the concrete selection decision. Moreover, the algorithm's transformation takes the candidates' observable qualifications as embodied in the raw score as a basis, thus arguably fulfilling the CJEU criteria just mentioned, particularly when paired with human-level overview of the re-ranking process (see also Craig and de B\'urca 2011, p. 915; Hacker 2018, p. 1181). From a legal point of view, it will be necessary but also sufficient for the final decision (credit allocation; admission) to be made by a human decision maker who has some leeway to overrule the results of the scoring process.\footnote{ Cf. CJEU, Case C-158/97, \emph{Badeck}, EU:C:2000:163, paras. 55 and 63 (concerning selection for training and interview).}  This, in turn, helps to ensure compliance with data protection law, too (cf.\ Art.~22(1) of the EU General Data Protection Regulation).

In principle, the more we slide $\theta$ towards group fairness, the less we can guarantee individual fairness between members of different groups, and vice versa. If courts or regulators give an indication concerning the exact trade-off between these two goals, it can be implemented with our algorithm. To prevent litigation, companies could, pre-emptively, set the parameter so that they minimize the risk of running afoul of affirmative action rules. Ultimately, however, it is important to note that automated technological fixes alone, such as algorithmic re-ranking, will not resolve the conflict between individual and group equality; rather, algorithmic fairness procedures, at least post-processing approaches, must be paired with a human review of re-ranking decisions to pass muster before affirmative action law.

\subsection{The Choice of $\theta$}

Within this legal framework, the choice of the \emph{substantive value} of $\theta$ is obviously crucial for our model. Ideally, the debate about what $\theta$ to choose for different situations of algorithmic decision making ought to be governed by a broad democratic discourse. Despite the formality of the model, its core features, we believe, can be discussed outside of a technocratic framework and may even become the subject of discussion in legislative bodies. As a contribution to this discourse, we would like to offer the following suggestions for the choice of an adequate $\theta$.

First, the more we can be sure that the raw scores, and the raw distributions, capture ground truth and are free from bias and other exogenous distortions, the less there is a need to transform the raw data into a fair representation; $\theta$ may then assume a relatively smaller value than in cases in which we suspect significant bias. Conversely, if we have reason to believe that the training data are ridden with bias, and if ground truth is not available, a high $\theta$ score seems attractive (see also Zafar et al. 2017, p. 2).

Second, since the choice of $\theta$ governs the trade-off between more group fairness (high $\theta$) and more individual fairness (low $\theta$), as discussed in the introduction and the model, we suppose that different $\theta$ values will be appropriate for different situations. In fact, it is the beauty of the model that it allows us to adapt to different areas of algorithmic decision making depending on whether we want to strengthen group or individual fairness. As a tentative suggestion, we could say in branches that deal with decisions that decisively shape the socioeconomic prospects and capabilities of applicants for years, selection procedures must be closer to statistical parity to avoid indirect discrimination. Hence, group fairness, i.e., statistical parity, could be a greater concern in areas that a) primarily form the basis of future life opportunities (e.g., high school or college admissions; high-value credits); or that b) directly relate to basic needs (such as access to housing; to justice; or to health insurance). Conversely, in branches with less implications for the long-term socioeconomic position of applicants, individual fairness could be a greater motivation. This may concern areas where, primarily, past performance or events are evaluated (e.g., job applications, perhaps also access to smaller lines of credit). Hence, regulators or courts could refer to our algorithm in establishing different quantitative discrimination thresholds (or corridors) for different sectors of the economy.

Clad in academic metaphors, group fairness could be more important in situations similar to the granting of scholarships (enabling future flourishing), and individual fairness for awards (recognizing past achievements). Evidently, many areas will combine elements of both patterns; this is precisely where the advantage of our model lies that allows for intermediate degrees combining individual and group fairness. 

We would also like to stress that some areas may necessitate yet other fairness metrics. For example, in criminal justice sentencing decisions, the different costs of false positives and false negative decisions for the concerned individuals suggests a focus on error rate balance, (see, for the controversy surrounding the COMPAS algorithm, Chouldechova 2017; Kleinberg et al. 2016; Berk et al. 2018). The error rate balance is a statistical measure at the group level, but is not equivalent to statistical parity. One important future extension of our model, therefore, will be to include this measure as well.

\subsection{Designing Procedures to Implement CFA$\theta$}

On the \emph{procedural} side, models of algorithmic fairness call for new types of regulatory implementations (Veale and Binns 2017; Hacker 2017). We are aware that this raises several challenges for real-world implementation. To illustrate them, let us briefly return to the example of credit scoring. Our model necessitates three separate procedural steps: first, the computation of the raw score predicting credit worthiness; second, its transformation into a fair score; and, third, the application of the decision rule to this fair representation.

We take the first step, the calculation of the raw scores, as a given. This is precisely where machine learning models enter the scene. However, we do want to note that, already at this step, bias minimization techniques could apply (see, e.g., Calders and Verwer 2010 and the brief overview in section ``Relation to Other Work"). For the second step, an impartial, trusted party may be needed to perform the transformation of the raw into the fair representation. Moreover, ideally based on legislative guidance, this party also needs to determine the value of $\theta$. One key policy question will be whether private companies should take on the role of an impartial party, as in financial auditing; or whether a government agency should be endowed with this task (see Tutt 2017; Wachter et al. 2017, p. 98, for a discussion). Both solutions raise questions of conflicts of interest, capture etc. After the transformation of the scores, the banks needs to know the fair distribution. The application of the decision rule, by the decision maker, based on the fair rather than the raw data seems straightforward to us. 

In the end, implementing a framework for any fair decision making model in the algorithmic context will provide quite a challenge in the real world. With our model, we hope to facilitate the most relevant trade-offs involved, and to make transparent the different design choices that policymakers face.

\section{Experiments}
In this final section, we implement CFA$\theta$ and experimentally evaluate its performance on two data sets. First, we use synthetic data to show that our approach can be used to gradually reduce disparate impact between different groups by increasing $\theta$. Second, we apply CFA$\theta$ to a real-world data set, the law school data set comprising admissions and performance data of law school students. 

A note on terminology: We do not distinguish between protected and non-protected groups, only between protected and non-protected features, and between (initially) advantaged and disadvantaged groups. In legal terms, equal treatment law generally forbids direct or indirect discrimination against \emph{any} group singled out by a protected feature (e.g., male or female). In this sense, all groups defined by protected criteria are protected. Of course, in affirmative action settings, one group is usually underrepresented, i.e., initially disadvantaged. This does not imply, however, that members of the advantaged group are non-protected. Rather, the legal issues surrounding affirmative action, for example in the doctrine of the Court of Justice of the EU (see Section~5), arise precisely from the fact that legal protection is granted even to members of advantaged protected groups under antidiscrimination law.
\subsection{Data Sets}
\subsubsection{Synthetic Data}
Our synthetic data set consists of raw scores for 100,000 individuals. For each individual, there are two categories of protected features: gender (0, 1) and ethnicity (0, 1, 2), yielding six groups in total (Group 1: [0,0]; Group 2: [0,1] etc.), of which three are disadvantaged by their score distributions (Figure~4). Raw scores are (almost) normally distributed integers with different means and standard deviations for each of the six resulting groups. The individual values range from 3 to 88 and are distributed at random among the individuals, subject to the normal distribution constraint. We can imagine these raw scores to be credit scores calculated by a credit score agency such as FICO in the US or SCHUFA in Germany.
\begin{figure}\label{fig1}
\caption{Raw scores for synthetic data}
\includegraphics[width=0.80\textwidth]{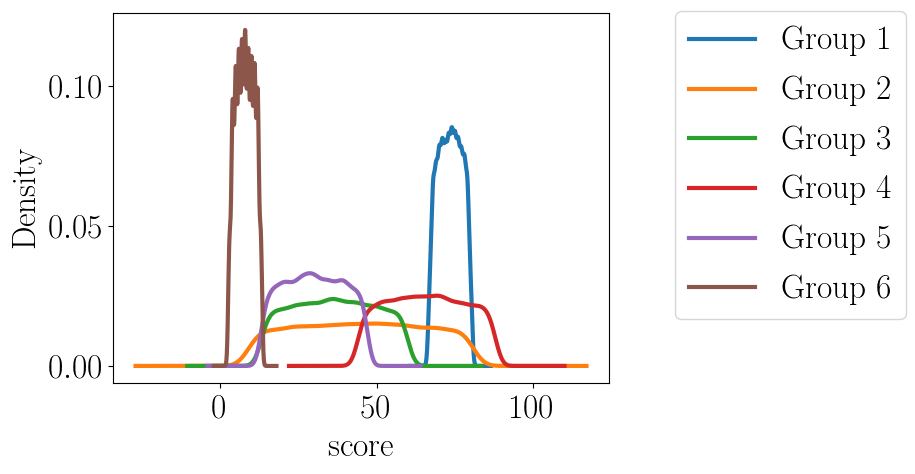}
\end{figure}
\subsubsection{Real-World Data: Law School Admissions}
The law school data set was first described by Wightman (1998). It was used for the study of differential bar exam passage rates between different ethnic groups, out of a concern that certain metrics or admissions policies disfavor ethnic minority groups, particularly black students. The data set includes anonymized data on 21,792 students in total and also comprises two categories of protected features: gender (male, female) and ethnicity (White, American Indian, Asian, Black, Hispanic, Mexican, Other). Women make up a little less than 44\% of all students, White students 84\%, American Indians around 0.5\%, Asians less than 4\%, Blacks less than 6\%, Hispanics around 2\%, Mexicans less than 2\%, and others around 1\%.
Furthermore, for each student, it records three non-protected features: the LSAT (law school admissions test) score; the undergraduate grade-point average; and the average grade at the end of the first year of law school (as a z-score: ZFYA). We use the LSAT score as the raw score for our analysis. LSAT scores range from 11 to 48 in the data set.
In an admissions setting, for example, a ranking model might be trained, using the law school data set as training data, to predict an equivalent of an LSAT score as the relevant admission score for the selection of the top k-ranked individuals. Disparities in the LSAT score would translate into disparities in the admissions policy. Hence, assigning fair LSAT scores remains an important task. Fairness measures do not differ much between gender groups, however. Therefore, we focus on the largest ethnic minority, black students, whose LSAT score distribution significantly and negatively differs from the other groups' score distributions (see Figure~5). 

\begin{figure}
\caption{LSAT scores by ethnicity}
\includegraphics[width=0.80\textwidth]{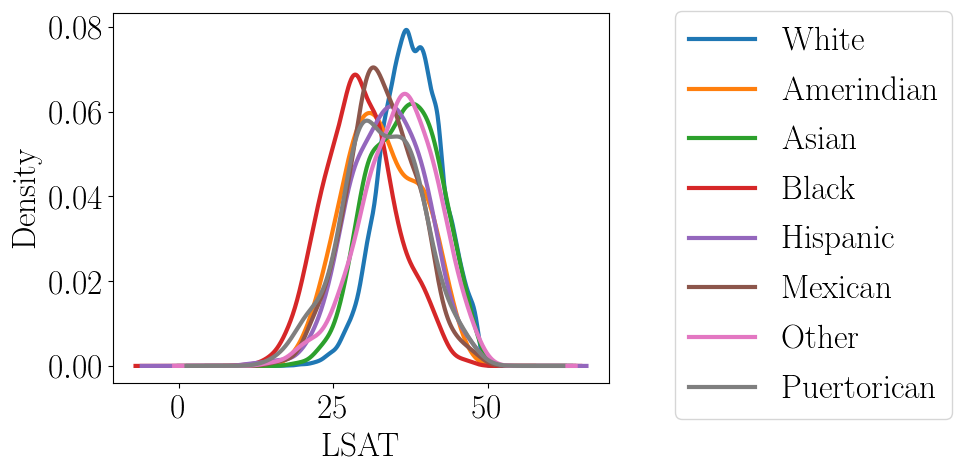}
\end{figure}

\subsection{Experimental Setting}
In the experiments, different parameters can be adjusted. First, the bin size determining the intervals in which scores are considered by the algorithm can be varied; for example, if one group consists of relatively few members, the bin size can be expanded so that a sufficient number of individuals fall into the bin. The default standard bin size for raw scores, chosen here, is 1. 

Second, most importantly, the $\theta$ parameter can be set anywhere between 0 and 1 for each group separately, depending on the desired trade-off between utility/individual fairness on the one hand and group fairness on the other hand. As explained in the model, a $\theta$ value of 0, in theory, produces fair scores that are identical to the raw scores, hence implementing a WYSIWYG worldview. Conversely, in theory, a $\theta$ score of 1 maps all raw scores fully onto the corresponding total barycenter scores. If chosen for all groups, it should therefore fully align distributions between the groups and implement a WAE worldview. Intermediate $\theta$ values between 0 and 1 correspond to an interpolation between these two extremes.

To achieve the transformation from raw to fair scores, we use the Python Optimal Transport library of Flamary \& Courty 2017. To experimentally measure the performance of CFA$\theta$, we implement two widely used performance metrics (see, e.g., Qin et al. 2010, at p. 360): precision at position k ($P@k$) and normalized discounted cumulative gain (NDCG).
$P@k$ describes by how much the fair ranking diverges from the raw ranking. Therefore, it approximately measures our individual fairness error (equation 2.9). More specifically, precision at position k evaluates the top k items of a ranking, based on whether they are relevant or irrelevant for the decision maker. We define an item (= individual) in the re-ranked fair score ranking as relevant if and only if the candidate was included in the top k in the raw ranking, otherwise as irrelevant. NDCG, by contrast, uses multiple levels of relevance to evaluate the top k-ranked items of a ranking (J{\"a}rvelin and Kek{\"a}l{\"a}inen 2002). For NDCG, gains are calculated directly using the individual scores. Depending on the rank, the gain is discounted logarithmically ($1 / \log_{2}(n + 1)$), so that gains by lower-ranked individuals count substantially less. This reflects the fact that they will often be considered less frequently by the decision maker. All discounted gains are summed up until position k and then normalized onto the interval $[0; 1]$ by a normalization coefficient reflecting an ideal ranking. In our context, for the purpose of measuring differences between the raw and the fair ranking, the raw score ranking functions as the ideal ranking. We note that this is a strong assumption and usually not realistic, as a trained IR model is not 100 \% correct with respect to the training data. We expect NDCG differences to be much smaller if one compares the differences between ground truth and predictions on the one hand and ground truth and fair score adjustment on the other. However, in our particular setting, we assume a 100 \% correct model to evaluate the maximum utility loss that our method can introduce to a ranking.

Furthermore, we evaluate fairness gains using an adaptation of a fairness measure which is widely used in the literature: the disparity measure (see Yang et al. 2018). The disparity measure is defined as the ratio of two proportions: the proportion of positively selected members of one protected but disadvantaged group compared to the proportion of that group in the entire data set. It makes particular sense to use this measure because a disparity measure of less than 0.8 (US) or 0.75 (EU) usually indicates indirect discrimination in antidiscrimination law (EEOC 2015, Section 4 D.). We operationalize this measure for our setting as in Yang et al., 2018, at p. 3. Both the raw and the fair scores can be used to rank the individuals. The top-k individuals are positively selected. What we measure directly is the share of positively selected individuals for each group at every k (Figures~12, 13 and 18). We then calculate at what cut-off point the disparity measure of the disadvantaged groups reaches 80\%.

Due to the differing number of individuals in the two data sets, we ran the evaluations of the experiments with a step size (for the k values) of 1,000 for the synthetic data and of 100 for the LSAT data. The findings do not change significantly, however, for other step sizes. 
\subsection{Results}
\subsubsection{Synthetic Data}
In our experiment with synthetic data, we perform three distinct fair score transformations from the raw scores, with $\theta$ values of 0, 0.5 and 1 for each group, respectively. As can be seen in the plots (Figure~5), the raw score distributions of the different groups differ significantly.

While $\theta$=0 reproduces this fact, increasing $\theta$ gradually has the group distributions converge toward the barycenter, see Figures~6--8. 

\begin{figure}
\caption{$\theta=0$}
\includegraphics[width=0.80\textwidth]{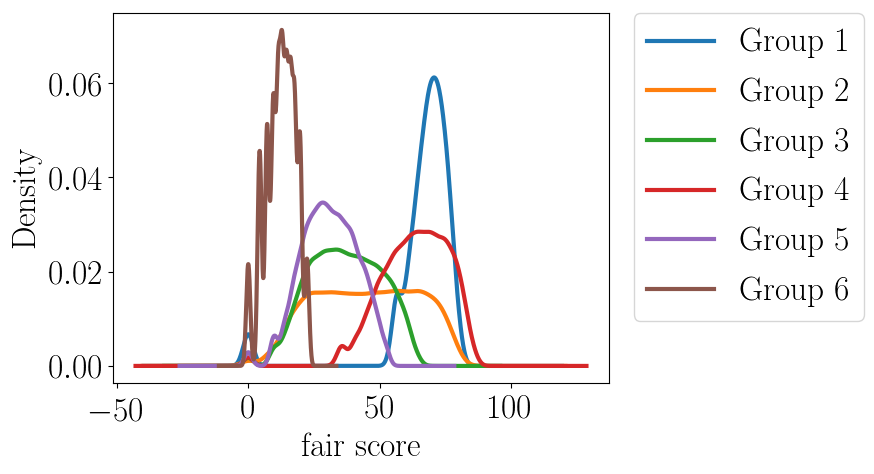}
\label{fig2}
\end{figure}
\begin{figure}
\caption{$\theta=0.5$}
\includegraphics[width=0.80\textwidth]{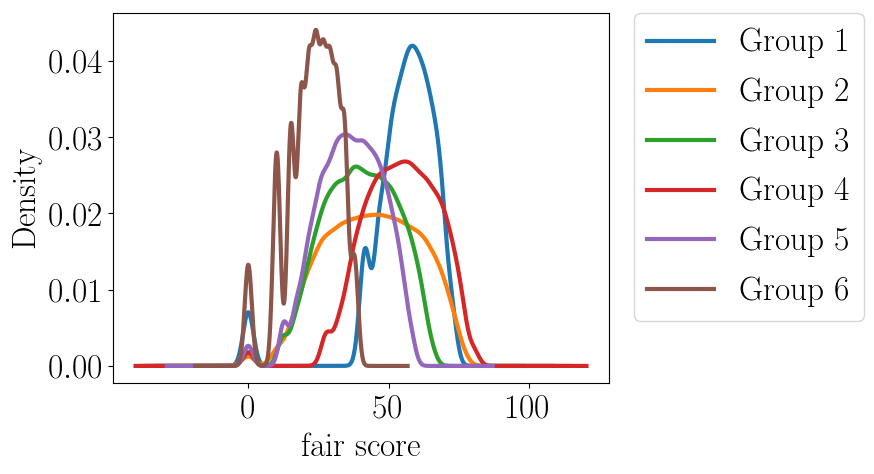}
\label{fig3}
\end{figure}
\begin{figure}
\caption{$\theta=1$}
\includegraphics[width=0.80\textwidth]{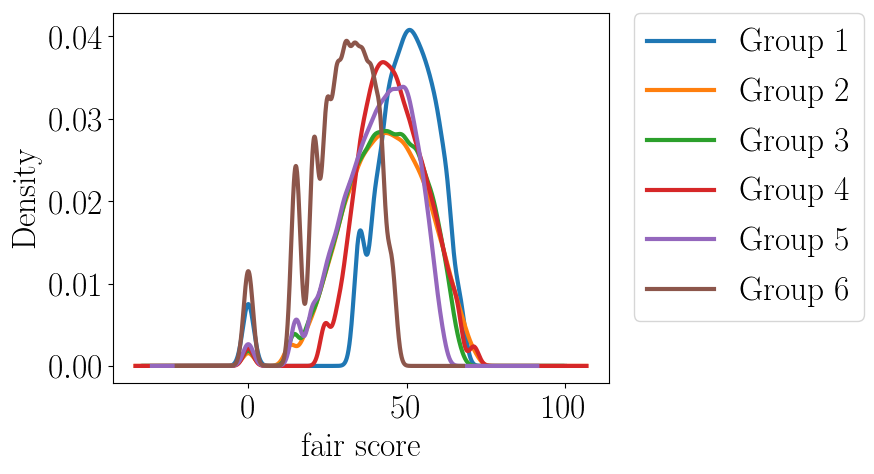}
\label{fig4}
\end{figure}

The performance evaluation reveals the following (Figures 9--11): 
\begin{figure}
\caption{Performance for $\theta=0$}
\includegraphics[width=0.80\textwidth]{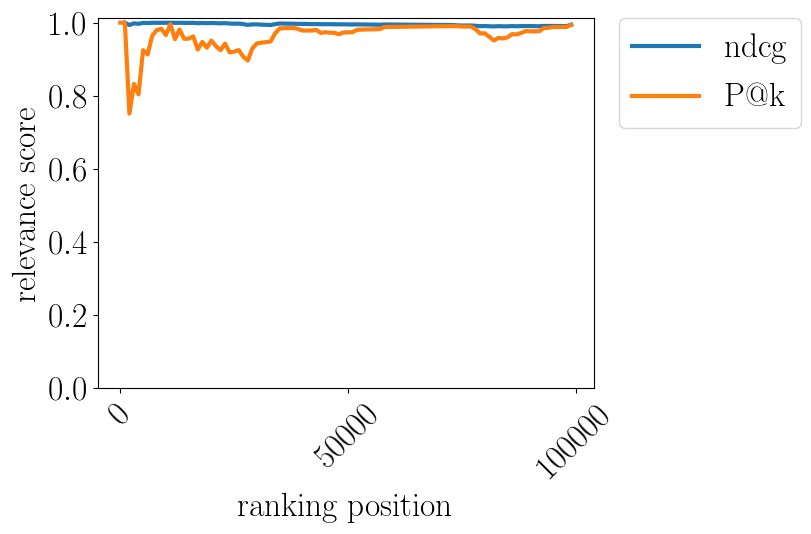}
\label{fig5}
\end{figure}
\begin{figure}
\caption{Performance for $\theta=0.5$}
\includegraphics[width=0.80\textwidth]{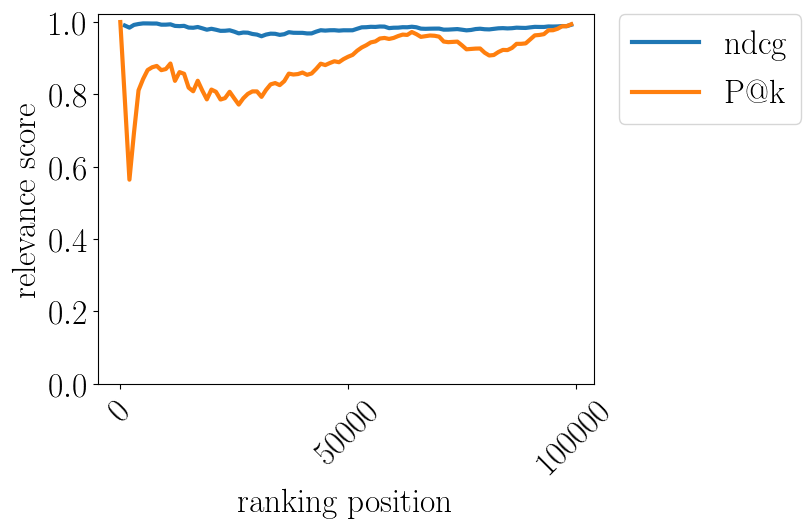}
\label{fig6}
\end{figure}
\begin{figure}
\caption{Performance for $\theta=1$}
\includegraphics[width=0.80\textwidth]{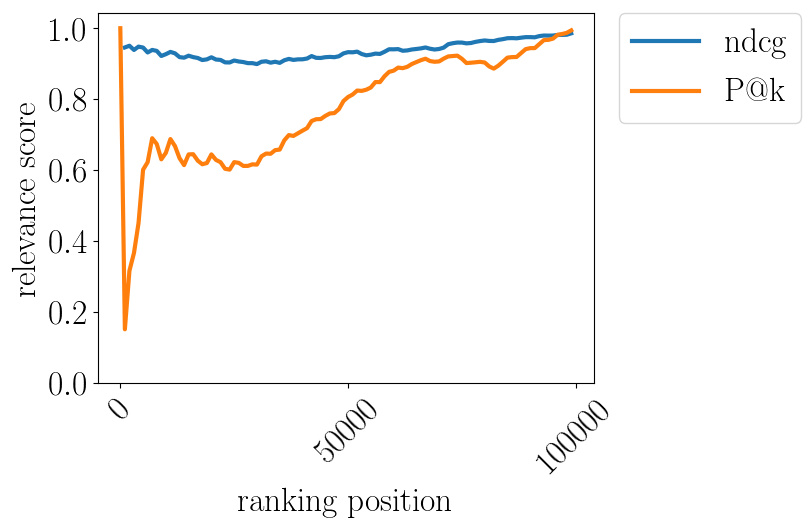}
\label{fig7}
\end{figure}
 for $\theta$=0, while the NDCG value, as expected, remains very close to 1, the $P@k$ value drops from about rank 1,000 to a minimum at 0.751 around rank 2,000. From a strictly mathematical viewpoint, $P@k$ should always equal 1 as the raw ranking should be fully preserved under $\theta=0$. While the $P@k$ value subsequently recovers and remains above 0.89 after rank 5,000, the drop in the upper part of the ranking reflects the limitations of our model discussed above: the algorithm has some difficulties in adjusting to groups that are either too small or exhibit too little variance within a certain interval. This implies, for practical purposes, that instead of using fair scores for $\theta=0$, the raw scores should be used.
The performance is much better for higher $\theta$ values. We lose only a maximum of 3.9 \% in NDCG performance for $\theta=0.5$ and of only 10.1 \% for $\theta=1$. The data therefore does not contradict our mathematical conclusion that the optimality of the transport preserves decision maker utility to a maximum. $P@k$ varies and drops substantially for both $\theta=0.5$ and $\theta=1$, which reflects the fact that the algorithm does re-rank individuals, testifying to an increase in what has here been termed individual fairness error. 
Simultaneously, the fairness evaluation (Figures 12 and 13)
\begin{figure}
\caption{Fairness evaluation for $\theta=0.5$}
\includegraphics[width=0.80\textwidth]{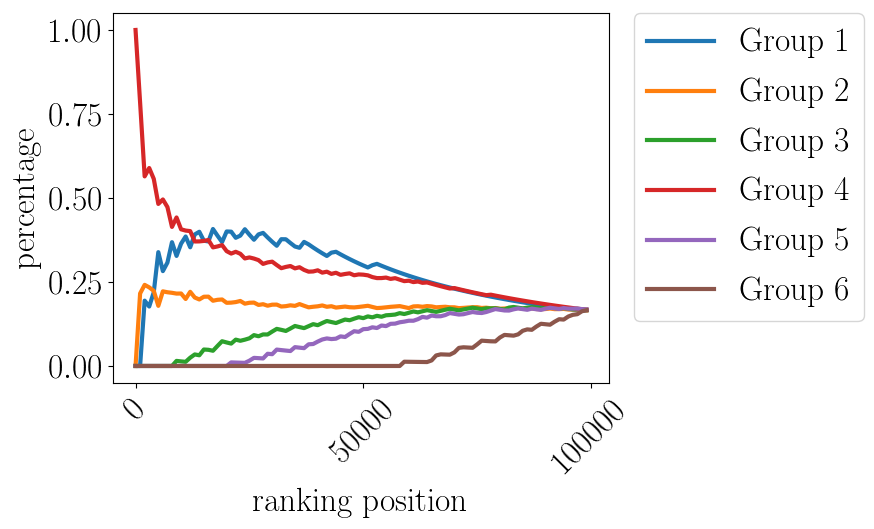}
\label{fig8}
\end{figure}
\begin{figure}
\caption{Fairness evaluation for $\theta=1$}
\includegraphics[width=0.80\textwidth]{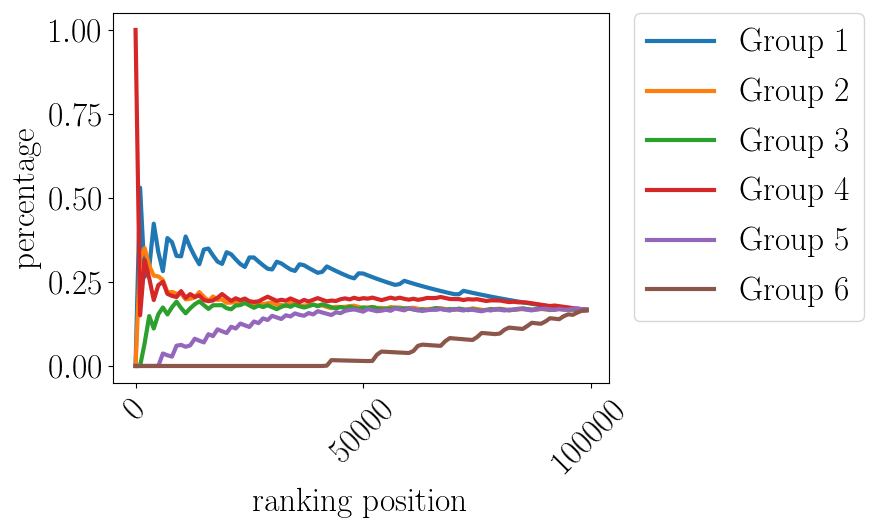}
\label{fig9}
\end{figure}
shows that already for $\theta=0.5$, but even more so for $\theta=1$, many more individuals of the initially disadvantaged groups are now listed in the higher parts of the ranking. In our data, we have three disadvantaged groups: Group 3, Group 5 and Group 6 (ranging from least to most disadvantaged). For $\theta=0$, they reach a disparity measure of 0.8 at the following ranking positions: from rank 53,000 on (Group 3); 64,000 (Group 5); and 95,000 (Group 6).\footnote{ We calculate these thresholds as follows: Each group has a share of 16.6 \% of the entire population. Hence, to reach a disparity measure of more than 0.8, the groups need to have a positive selection rate of at least 13.28 \%. Group 3 reaches this threshold from rank 53,000 on, with 13.42 \%; Group 5 at rank 64,000 with 13.5 \%; and Group 6 at rank 95,000 with 13.6 \%. For each of the other fairness evaluations, we calculate the thresholds correspondingly.} For $\theta=0.5$, the disparity threshold is reached earlier: at rank 46,000 (Group 3); 60,000 (Group 5); and 93,000 (Group 6). For $\theta=1$, the threshold is further lowered, quite substantially, to rank 5,000 (Group 3); 27,000 (Group 5); and 91,000 (Group 6). 

Therefore, the data shows that indeed, as expected, an increase in $\theta$ increases the individual fairness error (and hence decreases decision maker utility), but also increases group fairness. Again, we need to stress that the results imply that our algorithmic implementation of the model works well for large groups and data sets in which a large number of individuals are eventually considered for a decision. It is not suitable for decision making contexts in which only the first 10-100 individuals are accepted. This is due to a loss in precision of the optimal transport procedure as we move from the continuous mathematical framework to the necessarily approximative algorithmic implementation.
\subsubsection{LSAT}
The experiment with the LSAT data set is restricted to the $\theta$ values of 0 and 1. Qualitatively speaking, the results of the experiment with synthetic data are reproduced (Figures 14 and 15): the higher $\theta$ value leads to an increased individual fairness error (= decreased decision maker utility), but to an increase in group fairness.
\begin{figure}
\caption{LSAT data for $\theta=0$}
\includegraphics[width=0.80\textwidth]{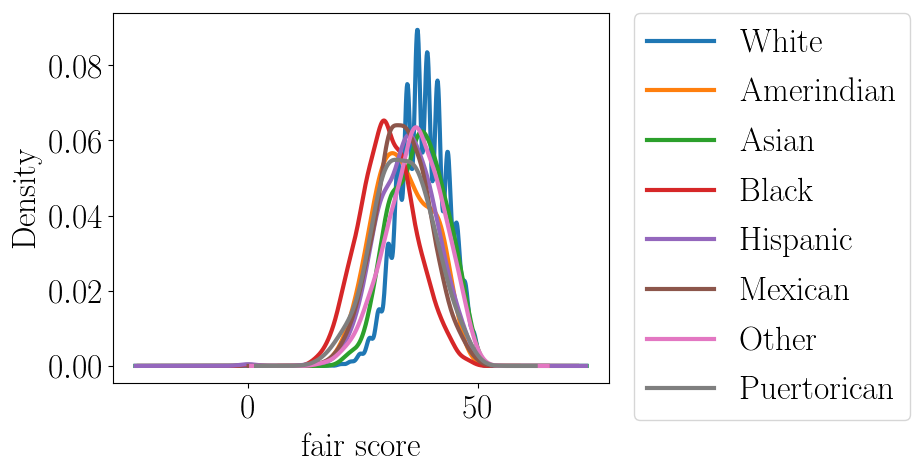}
\label{fig10}
\end{figure}
\begin{figure}
\caption{LSAT data for $\theta=1$}
\includegraphics[width=0.80\textwidth]{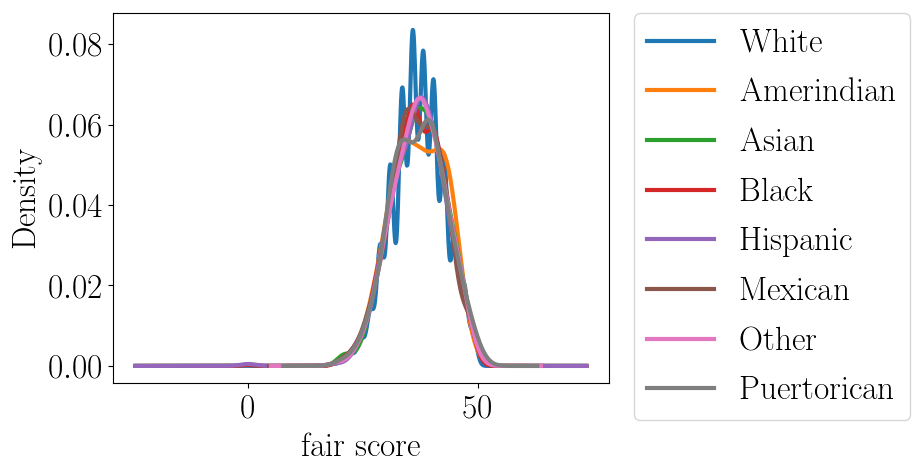}
\label{fig11}
\end{figure} 
However, three additional features should be noted. First, while the NDCG scores for both $\theta$ values remain very close to one, the $P@k$ values drop a little more than in the synthetic data set, the minimum under $\theta=0$ being reached at 0.717 for rank 600. Second, however, NDCG performance is significantly better than with synthetic data, with a maximum loss, for $\theta=1$, of 1.2 \% at rank 1,500 (Figures 16 and 17). 
\begin{figure}
\caption{Performance at $\theta=0$}
\includegraphics[width=0.80\textwidth]{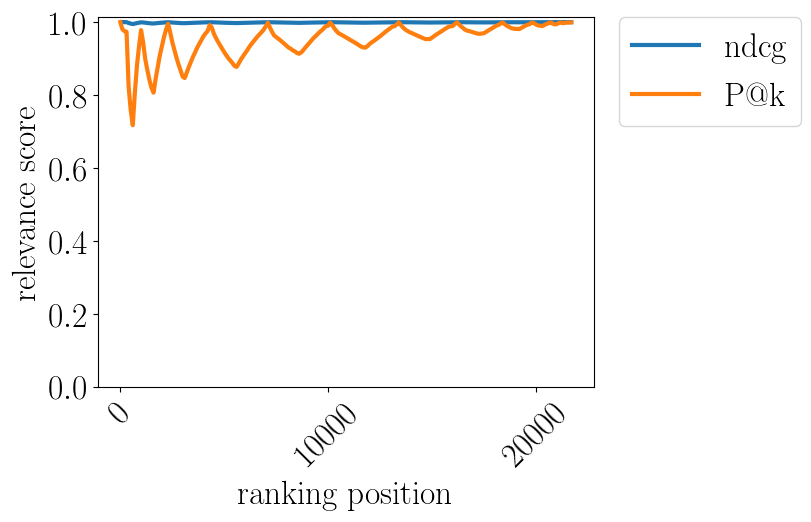}
\label{fig12}
\end{figure}
\begin{figure}
\caption{Performance at $\theta=1$}
\includegraphics[width=0.80\textwidth]{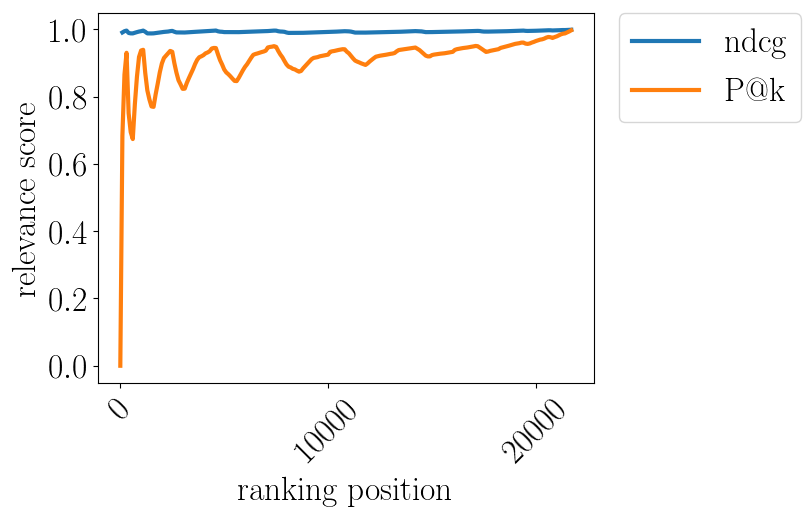}
\label{fig13}
\end{figure}

Third, the fairness evaluation (Figure 18)
\begin{figure}
\caption{Fairness evaluation for LSAT, $\theta=1$}
\includegraphics[width=0.80\textwidth]{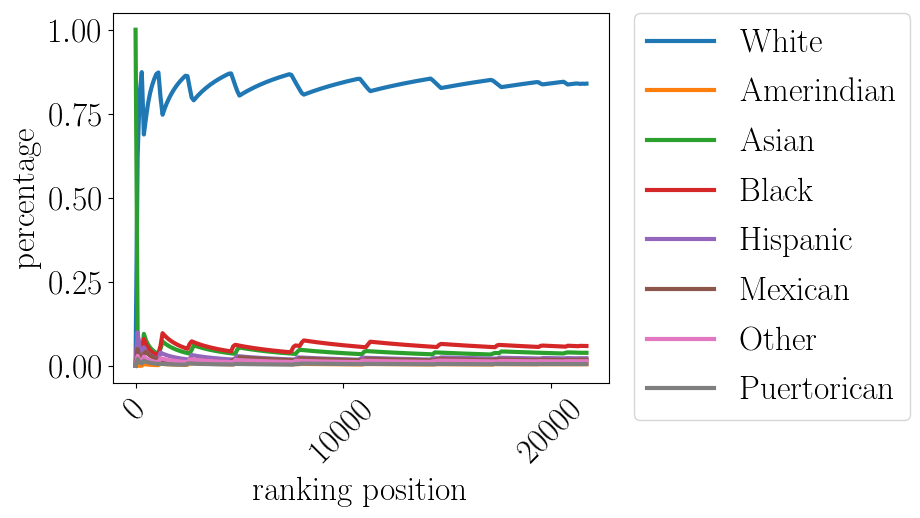}
\label{fig14}
\end{figure}
shows that disadvantaged groups do perform significantly better under $\theta=1$. The most disadvantaged group, black students, crosses the 0.8 disparity threshold for $\theta=0$ only at rank 21,300; for $\theta$=1, it is lifted above 0.8 already at rank 100, then drops again below that measure, to remain above 0.8 again for ranks 400-600, 1,200-6,500, and then from 7,600 on. Hence, in a typical law school class of 500 students, we would find disparate impact under $\theta=0$ (i.e., using the raw scores): the disparity measure is less than 0.2. However, for $\theta=1$, the disparity measure is around 1, averting disparate impact vis-\`a-vis black students.
\subsubsection{Running Time and Scalability}
The algorithm was implemented on a standard laptop computer. The running time for the re-ranking ranged between a maximum of 12.3 s for $\theta=1$ for the synthetic data and a minimum of 6.5 s for $\theta=1$ for the LSAT data set. As our synthetic data set contained 100,000 individuals, this shows that our algorithm scales well to data sets with large numbers of individuals.

\section{Conclusion}

This paper presents a new algorithm which harnesses optimal transport theory to maximize decision maker utility under fairness constraints. The algorithm facilitates the trade-off between different fairness measures (individual fairness error and group fairness) as well as different fairness worldviews (WYSIWYG and WAE). Since the degree of implementation of group fairness can be varied, the algorithm helps decision makers adapt their models to varying legal constraints in different situations. This finding is validated in a number of data-driven experiments. The key take-away from the experiments is that the method does not perform well if groups are small. However, when groups are large, it presents a fast, scalable and functional way to address issues of discrimination in a variety of decision making contexts, from credit decisions to insurance and large admissions decisions.


\end{document}